\documentclass[%
 reprint,
superscriptaddress,
 amsmath,amssymb,
aps,
pra,
]{revtex4-1}

\usepackage{amsthm}
\usepackage{graphicx}% Include figure files
\usepackage{dcolumn}% Align table columns on decimal point
\usepackage{bm}% bold math
\usepackage{enumerate}
\usepackage{algorithm}
\usepackage{algorithmic}

\usepackage{url}
\usepackage{float}

\usepackage{amstext}
\theoremstyle{remark}

\newtheorem{lemma}{\indent Lemma}

\UseRawInputEncoding
\begin{document}

\makeatletter
\newcommand{\ud}{\mathrm{d}}
\newcommand{\rmnum}[1]{\romannumeral #1}
\newcommand{\Rmnum}[1]{\expandafter\@slowromancap\romannumeral #1@}
\newcommand{\udots}{\mathinner{\mskip1mu\raise1pt\vbox{\kern7pt\hbox{.}}
        \mskip2mu\raise4pt\hbox{.}\mskip2mu\raise7pt\hbox{.}\mskip1mu}}
\makeatother

\preprint{APS/123-QED}

\title{Quantum discriminative canonical correlation analysis}

\author{Yong-Mei Li}
\affiliation{State Key Laboratory of Networking and Switching Technology, Beijing University of Posts and Telecommunications, Beijing, 100876, China}
\affiliation{State Key Laboratory of Cryptology, P.O. Box 5159, Beijing, 100878, China}
\author{Hai-Ling Liu}
\affiliation{State Key Laboratory of Networking and Switching Technology, Beijing University of Posts and Telecommunications, Beijing, 100876, China}
\author{Shi-Jie Pan}
\affiliation{State Key Laboratory of Networking and Switching Technology, Beijing University of Posts and Telecommunications, Beijing, 100876, China}
\author{Su-Juan Qin}
\email{qsujuan@bupt.edu.cn}
\affiliation{State Key Laboratory of Networking and Switching Technology, Beijing University of Posts and Telecommunications, Beijing, 100876, China}
\author{Fei Gao}
%\email{gaof@bupt.edu.cn}
\affiliation{State Key Laboratory of Networking and Switching Technology, Beijing University of Posts and Telecommunications, Beijing, 100876, China}
\author{Qiao-Yan Wen}
%\email{wqy@bupt.edu.cn}
\affiliation{State Key Laboratory of Networking and Switching Technology, Beijing University of Posts and Telecommunications, Beijing, 100876, China}

\date{\today}

\begin{abstract}
Discriminative Canonical Correlation Analysis (DCCA) is a powerful supervised feature extraction technique for two sets of multivariate data, which has wide applications in pattern recognition. DCCA consists of two parts: (i) mean-centering that subtracts the sample mean from the sample; (ii) solving the generalized eigenvalue problem. The cost of DCCA is expensive when dealing with a large number of high-dimensional samples. To solve this problem, here we propose a quantum DCCA algorithm. Specifically, we devise an efficient method to compute the mean of all samples, then use block-Hamiltonian simulation and quantum phase estimation to solve the generalized eigenvalue problem. Our algorithm achieves a polynomial speedup in the dimension of samples under certain conditions over its classical counterpart.
\end{abstract}

\pacs{Valid PACS appear here}
\maketitle

\section{Introduction}

As a promising new computing paradigm, quantum computing has shown enormous power over classical computing due to its inherent parallelism and entanglement~\cite{PS1994,GL1996}. In recent years, a series of quantum algorithms have been proposed with significant speedups compared with their classical counterparts, such as classification~\cite{RML,QC1,QC2,jinan}, linear regression~\cite{LR1,LR2,LR3}, dimensionality reduction~\cite{QPCA,QLDA,DYXL,pan2022}, matrix computation~\cite{HHL,DM2018,wan2018,liu}. Designing quantum algorithms to accelerate the corresponding classical ones has appeared as a remarkable emerging direction in the field of quantum computing.

Feature extraction for pairs of multivariate data is an indispensable part of multimodal recognition and information fusion, which aims at extracting the feature pairs from two groups of feature vectors~\cite{DCCA,SCCA}. Canonical Correlation Analysis (CCA)~\cite{CCA1,CCA2} and Partial Least Squares (PLS)~\cite{PLS} are widely-used feature extraction techniques for pairs of multivariate data. For multimodal recognition task, CCA and PLS are unsupervised, that is they do not utilize the class information of the samples, resulting in the constraint of the recognition performance. To solve this problem, Sun et al. proposed a novel combined feature extraction technique known as Discriminative Canonical Correlation Analysis (DCCA)~\cite{DCCA}, which takes the advantage of class label information to find pairs of pairwise projection vectors such that the within-class correlations are maximized and simultaneously the between-class similarities are minimized. DCCA can be described as two steps: (i) mean-centering that subtracts the sample mean from the sample; (ii) solving the generalized eigenvalue problem. Since the time complexity of DCCA depends polynomially on the number and the dimension of samples, it is computationally expensive when dealing with a large number of high-dimensional samples. Therefore, it would be of great interest to design the quantum algorithm for DCCA.

Recently, Koide-Majima et al. presented a quantum-inspired CCA (qiCCA) algorithm~\cite{qiCCA}, which achieves an exponential acceleration on the dimension of samples compared with the original algorithm. After that, Hou et al. suggested a Quantum PLS (QPLS) regression algorithm~\cite{QPLS}, which realizes an exponential speedup in the number of samples and their dimension over its classical counterpart. However, the above algorithms are still unsupervised and can not break through the limitation of the recognition performance in multimodal recognition. It is desirable to design a quantum algorithm for DCCA. DCCA, which utilizes the class label information, is fundamentally different from CCA and PLS, and it seems  infeasible to obtain a quantum DCCA algorithm directly based on the above qiCCA and QPLS.

In this paper we propose a Quantum DCCA (QDCCA) algorithm. First, we propose an efficient method, called Mean Estimation (ME), to compute the mean of all samples. The basic idea is inspired by inner product estimation in Ref.~\cite{qmeans2019}, which uses the amplitude estimation circuit~\cite{BHMT} to compute the inner product of two quantum states and then boosts the probability of success by median estimation~\cite{WKS2015}. We take advantage of the quantum state preparation technique to generalize this method to compute the mean of all elements in a row of any real matrix. The mean-centering of DCCA can be realized by ME and Quantum Multiply-Adder (QMA)~\cite{QMA,QMA2}. Then, inspired by Ref.~\cite{GEP3}, we use block-Hamiltonian simulation~\cite{ICALP2019} and quantum phase estimation to solve the generalized eigenvalue problem. In Ref.~\cite{GEP3}, with the help of the block-encoding technique~\cite{GSL,CGJ,ICALP2019}, Shao et al. put forward a novel quantum algorithm for solving the generalized eigenvalue problem $Dw =\lambda Ew$ when $D$ is Hermitian and $E$ is Hermitian positive definite. Here we design the block-encodings of the related matrices to give a variant of Shao et al.'s algorithm such that $Dw =\lambda Ew$ can be solved when $D$ is Hermitian and $E$ is Hermitian positive semidefinite. Based on this, the generalized eigenvalue problem of DCCA can be solved. As a result, the QDCCA algorithm achieves a polynomial speedup in the dimension of samples under certain conditions over the classical algorithm.

This paper is organized as follows. We review the classical DCCA in Sec.~\ref{Sec:DCCA}. In Sec.~\ref{Sec:QDCCA}, we propose the QDCCA algorithm in Sec.~\ref{subsec:QDCCA} and analyze its complexity in Sec.~\ref{subsec:CA}. The conclusion is given in Sec.~\ref{Sec:con}.

\section{Review of the classical DCCA}\label{Sec:DCCA}

Given $n$ pairs of original pairwise samples $\{(\textbf{a}_i,\textbf{b}_i)\}_{i=1}^n\in R^p\times R^q$ coming from $c$ classes, DCCA aims  to find $d$ pairs of pairwise projection vectors under specified conditions, where $d$ is a prespecified parameter and satisfies the constraints $d\leq \min(p,q)$ and $d\leq c$~\cite{DCCA}. Let $A=[\textbf{a}_1,\textbf{a}_2,...,\textbf{a}_n]\in R^{p\times n}$, $B=[\textbf{b}_1,\textbf{b}_2,...,\textbf{b}_n]\in R^{q\times n}$ denote the original data matrices and $M=\left(\begin{matrix}
   A  \\
   B\\
  \end{matrix}
  \right)$. Here we describe DCCA as the following two steps.

\emph{\textbf{Step 1.} Mean-centering}.

In this step, the Mean-Subtraction (MS) should be used to make all sample features have zero mean. The details are as follows.

(1) Compute two means of all samples:
\begin{eqnarray}
\begin{aligned}
\bar{\textbf{a}}=\frac{\sum_i\textbf{a}_i}{n},\ \  \bar{\textbf{b}}=\frac{\sum_i\textbf{b}_i}{n}.
\end{aligned}
\end{eqnarray}

(2) Each sample subtract the corresponding mean, then the centralized pairwise samples $\{(\textbf{x}_i,\textbf{y}_i)\}_{i=1}^n$ can be obtained where  $\textbf{x}_i :=\textbf{a}_i-\bar{\textbf{a}}$ and $\textbf{y}_i :=\textbf{b}_i-\bar{\textbf{b}}$.

\emph{\textbf{Step 2.} Solve the generalized eigenvalue problem}.

The first pair of pairwise projection vectors $(\textbf{w}_x,\textbf{w}_y)$ of DCCA can be formulated as the following optimization problem:
\begin{eqnarray}
\begin{aligned}
&\mathop{\max}_{\textbf{w}_x,\textbf{w}_y}\ \textbf{w}_x^TXCY^T\textbf{w}_y\\
&s.t.\ \textbf{w}_x^TXX^T\textbf{w}_x=1,\textbf{w}_y^TYY^T\textbf{w}_y=1,
\end{aligned}
\end{eqnarray}
where $X=[\textbf{x}_1,\textbf{x}_2,...,\textbf{x}_n]$, $Y=[\textbf{y}_1,\textbf{y}_2,...,\textbf{y}_n]$ are the centralized date matrices, $C=\text{diag}(1_{n_1},1_{n_2},...,1_{n_c})$, $1_{n_i}$ is an $n_i\times n_i$ matrix with all ones, $n_i$ denotes the number of pairwise samples in the $i$th class, $i=1,2,...,c$ and $\sum_in_i=n$.

Using the Lagrangian multiplier technique, it can be transformed into the generalized eigenvalue problem:
\begin{eqnarray}\label{equ:GE}
\begin{aligned}
\hspace{-5mm}%{}里面负数为左移，正为右移
\left(
 \begin{matrix}
    & XCY^T  \\
   YCX^T &   \\
  \end{matrix}
  \right)
  \left (
 \begin{matrix}
   \textbf{w}_x\\
    \textbf{w}_y\\
  \end{matrix}
  \right  )=\lambda
  \left (
 \begin{matrix}
    XX^T & \\
  & YY^T  \\
  \end{matrix}
  \right  )
  \left (
 \begin{matrix}
   \textbf{w}_x\\
    \textbf{w}_y\\
  \end{matrix}
  \right ). \
\end{aligned}
\end{eqnarray}
The generalized eigenvectors corresponding to the first $d$ largest generalized eigenvalues $\lambda$ are exactly the $d$ pairs of pairwise projection vectors of DCCA.

Let $E^{\frac{1}{2}}
  \textbf{w}=\bf{v}$ where $E :=\text{diag}(XX^T,YY^T)$ and $\textbf{w} :=\left(\begin{matrix}
   \textbf{w}_x\\
    \textbf{w}_y\\
  \end{matrix}\right)$, we can reduce equation (\ref{equ:GE}) to a Hermitian eigenvalue problem:
\begin{eqnarray}
\begin{aligned}
H\bf{v}=\lambda \bf{v},
\end{aligned}
\end{eqnarray}
where $H :=E^{-\frac{1}{2}}DE^{-\frac{1}{2}}$ and $D :=\left(
 \begin{matrix}
    & XCY^T  \\
   YCX^T &   \\
  \end{matrix}
  \right)$.

Once the eigenvectors $\{\textbf{v}_i\}_{i=1}^d$ corresponding to the first $d$ largest eigenvalues of $H$ are obtained, we can get $\{\textbf{w}_i|\textbf{w}_i=E^{-\frac{1}{2}}\textbf{v}_i\}_{i=1}^d$ after postprocessing.

The time complexity of DCCA is $O(n(p+q)+(p+q)^3)$ where $O(n(p+q))$ comes from \emph{\textbf{Step 1}} and $O((p+q)^3)$ comes from \emph{\textbf{Step 2}}.

\section{quantum algorithm for DCCA}\label{Sec:QDCCA}

In this section we first present the QDCCA algorithm in Sec.~\ref{subsec:QDCCA}, then analyze its complexity in Sec.~\ref{subsec:CA}.

We start with some notations that will be useful throughout the paper. Note that we can rewrite $A$ as $A=[A^1,A^2,...,A^c]$, where the submatrix $A^i\in R^{p\times n_i}$ is the original data matrix of the $i$th class, $i=1,2...,c$. Similarly, $B=[B^1,B^2,...,B^c]$, $X=[X^1,X^2,...,X^c]$ and $Y=[Y^1,Y^2,...,Y^c]$. With such representations, we define the matrix $\mathcal{A}$ as $\mathcal{A}= [\mathcal{A}^1,\mathcal{A}^2,...,\mathcal{A}^c]$ where the submatrix $\mathcal{A}^i=[A^i,\textbf{0},...,\textbf{0}]\in R^{p\times n'}$, $\textbf{0}$ is a vector with all zeros and $n'=\max_i {n_i}$. Let $\mathcal{B}= [\mathcal{B}^1,\mathcal{B}^2,...,\mathcal{B}^c]$ where $\mathcal{B}^i=[B^i,\textbf{0},...,\textbf{0}]\in R^{q\times n'}$ and let $\mathcal{M}=\left(\begin{matrix}
   \mathcal{A} \\
   \mathcal{B} \\
  \end{matrix}\right)$.

In our quantum algorithm, to apply quantum phase estimation to reveal the eigenvalues of $H$, we must be able to realize $e^{iHt}$. To achieve it, we first analyze the structure of $H$ and find that $D$ can be rewritten as
\begin{eqnarray}
\begin{aligned}
D=\left(\begin{matrix}
   \mathbb{X}\mathbb{X}^T & \mathbb{X}\mathbb{Y}^T  \\
   \mathbb{Y}\mathbb{X}^T & \mathbb{Y}\mathbb{Y}^T \\
  \end{matrix}
  \right)-
  \left(\begin{matrix}
   \mathbb{X}\mathbb{X}^T & 0  \\
 0& \mathbb{Y}\mathbb{Y}^T \\
  \end{matrix}
  \right) :=J-K,
\end{aligned}
\end{eqnarray}
where
  $\mathbb{X} :=[\sum_j X^1_{*,j},\sum_j X^2_{*,j},...,\sum_j X^c_{*,j}]\in R^{p\times c}$, $\mathbb{Y} :=[\sum_j Y^1_{*,j},\sum_j Y^2_{*,j},...,\sum_j Y^c_{*,j}]\in R^{q\times c}$, and $X^i_{*,j}$, $Y^i_{*,j}$ are the $j$th column of $X^i$, $Y^i$ respectively, $i=1,2,...,c$. Then $H=E^{-\frac{1}{2}}JE^{-\frac{1}{2}}-E^{-\frac{1}{2}}KE^{-\frac{1}{2}}$. Using the block-encoding technique~\cite{GSL,CGJ,ICALP2019}, once the block-encodings of $E$, $J$, $K$ are implemented, the linear combination of block-encoded matrices allows us to construct the block-encoding of $H$, and then we can use the block-Hamiltonian simulation to realize $e^{iHt}$. However, it is not easy to create the block-encodings of $E$, $J$, $K$ directly. Fortunately, the matrices $E$, $J$, $K$ are all positive-semidefinite, we can prepare the density operators $\rho_E :=\frac{E}{tr(E)}$, $\rho_J :=\frac{J}{tr(J)}$, $\rho_K :=\frac{K}{tr(K)}$ and create corresponding block-encodings to realize the block-encoding of $\tilde{H} :=\rho_E^{-\frac{1}{2}}\rho_J\rho_E^{-\frac{1}{2}}-\rho_E^{-\frac{1}{2}}\rho_K\rho_E^{-\frac{1}{2}}=\frac{tr(E)}{tr(J)}H$. It means that we can realize $e^{iHt}$ by implementing $e^{i\tilde{H}t}$.

Our quantum algorithm is divided into four steps: (1) preparing the density operators $\rho_E$, $\rho_J$ and $\rho_K$; (2) designing the block-encoding of $\tilde{H}$; (3) estimating the eigenvalues of $H$ according to quantum phase estimation, and then searching the first $d$ largest eigenvalues of $H$ to get the corresponding eigenvectors $\{|\textbf{v}_i\rangle\}_{i=1}^d$; (4)postprocessing: use quantum matrix inversion technique to get $\{|\textbf{w}_{i}\rangle||\textbf{w}_{i}\rangle \propto E^{-\frac{1}{2}}|\textbf{v}_i\rangle\}_{i=1}^d$.

\subsection{Algorithm}\label{subsec:QDCCA}

Assume that the matrices $M$, $\mathcal{M}$ and the vector $\textbf{c} :=(n_1,n_2,...,n_c)^T$ are stored in Quantum Random Access Memory (QRAM)~\cite{QRAM} which allows
the following mappings to be performed in times $O[\log (n(p+q))]$, $O[\log (cn'(p+q))]$ and $O(\log c)$, respectively.
\begin{eqnarray}
\begin{aligned}
\textbf{O}_M: |i\rangle|j\rangle|0\rangle\rightarrow |i\rangle|j\rangle|M_{ij}\rangle,
\end{aligned}
\end{eqnarray}
\begin{eqnarray}
\begin{aligned}\label{equ:Mc}
\textbf{O}_{\mathcal{M}}: |i\rangle|j\rangle|0\rangle\rightarrow |i\rangle|j\rangle|\mathcal{M}_{ij}\rangle,
\end{aligned}
\end{eqnarray}
\begin{eqnarray}
\begin{aligned}
\textbf{O}_c: |i\rangle|0\rangle\rightarrow |i\rangle|n_i\rangle,
\end{aligned}
\end{eqnarray}
where $M_{ij}$ and $\mathcal{M}_{ij}$ are the $(i,j)$-entries of $M$ and $\mathcal{M}$ respectively.

The following lemma is necessary for our QDCCA algorithm, which is a variant of Lemma A.10 in Supplementary Material of Ref.~\cite{qmeans2019}.
\begin{lemma}(mean estimation).\label{lem:ME}
Assume that the matrix $L\in R^{d_1\times d_2}$ is stored in a QRAM, that is the unitary $\textbf{O}_L: |i\rangle|j\rangle|0\rangle\rightarrow|i\rangle|j\rangle|L_{ij}\rangle$ can be performed in time $O(\log (d_1d_2))$. For any $\Delta>0$ and $\epsilon>0$, there exists a quantum algorithm that computes in time $O(\frac{\max_{ij}|L_{ij}|\log (d_1d_2)\mathrm{log}\frac{1}{\Delta}}{\epsilon})$,
\begin{eqnarray}
\begin{aligned}
U_{mean}: |i\rangle|0\rangle\rightarrow|i\rangle|\bar{L}_{i,*}\rangle,
\end{aligned}
\end{eqnarray}
with  probability at least $1-2\Delta$, where $\bar{L}_{i,*}$ is the mean of all elements in the $i$th row of matrix $L$ and $\epsilon$ is the error of $\bar{L}_{i,*}$.
\end{lemma}
\begin{proof}
See Appendix~\ref{app:proof}.
\end{proof}

Now we detail the process of the QDCCA algorithm.

\emph{\textbf{Step 1.} Prepare the density operators $\rho_E$, $\rho_J$ and $\rho_K$. }

The matrices $E$, $J$, $K$ can be decomposed into $E=\acute{E}\acute{E}^T$, $J=\acute{J}\acute{J}^T$ and $K=\acute{K}\acute{K}^T$ respectively, where
  $\acute{E} :=\text{diag}(X,Y)$,
  $\acute{J} :=\left(\begin{matrix}
   \mathbb{X}    \\
    \mathbb{Y}   \\
  \end{matrix}
  \right)$,
    $\acute{
    K} :=\text{diag}(\mathbb{X},\mathbb{Y})$. According to the construction of $E$, $J$, $K$, we find that $\rho_E$, $\rho_J$ and $\rho_K$ can be obtained by tracing out the first register from the following three quantum states respectively: $|\psi_E\rangle=\frac{1}{\|\acute{E}\|_F}\sum_{j=1}^{2n}\sum_{i=1}^{p+q}\acute{E}_{ij}|j\rangle|i\rangle$, $|\psi_J\rangle=\frac{1}{\|\acute{J}\|_F}\sum_{i=1}^{c}\sum_{k=1}^{p+q}\acute{J}_{ki}|i\rangle|k\rangle$ and $|\psi_K\rangle=\frac{1}{\|\acute{K}\|_F}\sum_{i=1}^{2c}\sum_{k=1}^{p+q}\acute{K}_{ki}|i\rangle|k\rangle$. That is to say, once these states are obtained, we can get $\rho_E$, $\rho_J$ and $\rho_K$.

We now elaborate how to prepare the states $|\psi_E\rangle$, $|\psi_J\rangle$ and $|\psi_K\rangle$. For simplicity, here we use $I_a$ to represent the identity operator acting on the $a$th register. The details are as follows.

(1) Prepare the state $|\psi_E\rangle$.

According to the construction of matrices $M$ and $\acute{E}$, we know
\begin{eqnarray}
    \begin{aligned}
    \acute{E}_{ij}=
            \begin{cases}
            M_{ij}-\bar{M}_{i,*}, & 1\leq i\leq p,\ 1 \leq j\leq n;\\
M_{i,j-n}-\bar{M}_{i,*},&p+1\leq i\leq p+q,\\& \ n+1\leq j\leq 2n; \\
0,&\text{otherwise}.\\
            \end{cases}
    \end{aligned}
    \end{eqnarray}
We can use \emph{Lemma}~\ref{lem:ME} to calculate $\bar{M}_{i,*}$, then combine with QMA to obtain the state $\frac{1}{\sqrt{2n}}\sum_{j=1}^{2n}|j\rangle\frac{1}{\sqrt{p+q}}\sum_{i=1}^{p+q}|i\rangle|\acute{E}_{ij}\rangle$. Afterwards, it is possible to perform controlled rotation~\cite{CR} and fixed-point quantum search~\cite{fix2005,fix2014} to obtain $|\psi_E\rangle$. The specific process is as follows.

(1.1) Prepare the initial state
\begin{eqnarray}
\begin{aligned}
\frac{1}{\sqrt{2n}}\sum_{j=1}^{2n}|j\rangle_1\frac{1}{\sqrt{p+q}}\sum_{i=1}^{p+q}|i\rangle_2|0\rangle_3|0\rangle_4|0\rangle_5^{\otimes 2}.
\end{aligned}
\end{eqnarray}

(1.2) By \emph{Lemma}~\ref{lem:ME}, with a given $\textbf{O}_M$, we can implement $U_{mean}$ on the second and third registers to obtain
\begin{eqnarray}
\begin{aligned}
\frac{1}{\sqrt{2n}}\sum_{j=1}^{2n}|j\rangle_1\frac{1}{\sqrt{p+q}}\sum_{i=1}^{p+q}|i\rangle_2|\bar{M}_{i,*}\rangle_3|0\rangle_4|0\rangle_5^{\otimes 2}.
\end{aligned}
\end{eqnarray}

(1.3) Perform $U=\sum_{j=1}^n|j\rangle\langle j|\otimes \sum_{i=1}^p|i\rangle\langle i|\otimes (I\otimes I) + \sum_{j=1}^n|j\rangle\langle j|\otimes \sum_{i=p+1}^{p+q}|i\rangle\langle i|\otimes (I\otimes X)+\sum_{j=n+1}^{2n}|j\rangle\langle j|\otimes \sum_{i=1}^{p}|i\rangle\langle i|\otimes (X\otimes I)+\sum_{j=n+1}^{2n}|j\rangle\langle j|\otimes \sum_{i=p+1}^{p+q}|i\rangle\langle i|\otimes (X\otimes X)$ on the first, second and fifth registers to get
\begin{eqnarray}
\begin{aligned}
&\frac{1}{\sqrt{2n}}\bigg[\sum_{j=1}^{n}|j\rangle_1\frac{1}{\sqrt{p+q}}\bigg(\sum_{i=1}^{p}|i\rangle_2|\bar{M}_{i,*}\rangle_3|0\rangle_4|00\rangle_5\\
&+\sum_{i=p+1}^{p+q}|i\rangle_2|\bar{M}_{i,*}\rangle_3|0\rangle_4|01\rangle_5\bigg)\\
&+\sum_{j=n+1}^{2n}|j\rangle_1\frac{1}{\sqrt{p+q}}\bigg(\sum_{i=1}^{p}|i\rangle_2|\bar{M}_{i,*}\rangle_3|0\rangle_4|10\rangle_5\\
&+\sum_{i=p+1}^{p+q}|i\rangle_2|\bar{M}_{i,*}\rangle_3|0\rangle_4|11\rangle_5\bigg)\bigg],
\end{aligned}
\end{eqnarray}
where $X$ is the Pauli-$X$ gate and $I$ is the identity operator.

(1.4) Given a unitary $U_M: |i\rangle|j\rangle|0\rangle\rightarrow |i\rangle|j\rangle|M_{i,j-n}\rangle$ for $j=n+1,n+2,...,2n$, the following state can be obtained by performing
$\textbf{O}_M\otimes |00\rangle \langle 00|+U_M\otimes |11\rangle \langle 11|+I_{1,2,4}\otimes(|01\rangle \langle 01|+|10\rangle \langle 10|)$ on the first, second, fourth and fifth registers.
\begin{eqnarray}
\begin{aligned}
&\frac{1}{\sqrt{2n}}\bigg[\sum_{j=1}^{n}|j\rangle_1\frac{1}{\sqrt{p+q}}\bigg(\sum_{i=1}^{p}|i\rangle_2|\bar{M}_{i,*}\rangle_3|M_{ij}\rangle_4|00\rangle_5\\
&+\sum_{i=p+1}^{p+q}|i\rangle_2|\bar{M}_{i,*}\rangle_3|0\rangle_4|01\rangle_5\bigg)\\
&+\sum_{j=n+1}^{2n}|j\rangle_1\frac{1}{\sqrt{p+q}}\bigg(\sum_{i=1}^{p}|i\rangle_2|\bar{M}_{i,*}\rangle_3|0\rangle_4|10\rangle_5\\
&+\sum_{i=p+1}^{p+q}|i\rangle_2|\bar{M}_{i,*}\rangle_3|M_{i,j-n}\rangle_4|11\rangle_5\bigg)\bigg].
\end{aligned}
\end{eqnarray}

We now detail the $U_M$. For $j=n+1,n+2,...,2n$, we can first construct the unitary
$|i\rangle|j\rangle|0\rangle|0\rangle\rightarrow |i\rangle|j\rangle|j-n\rangle|0\rangle$ based on a unitary $|j\rangle|0\rangle\rightarrow|j\rangle|f(j)\rangle$ where $f(j)=j-n$, then apply $\textbf{O}_{M}$ to the first, third and fourth registers, the state $|i\rangle|j\rangle|M_{i,j-n}\rangle$ can be obtained after uncomputing the third register.

(1.5) Implement $U_{QMA}\otimes (|00\rangle \langle 00|+|11\rangle \langle 11|)+I_{3,4}\otimes(|01\rangle \langle 01|+|10\rangle \langle 10|)$ on the third, fourth and fifth registers to get
\begin{eqnarray}
\begin{aligned}
&\frac{1}{\sqrt{2n}}\bigg[\sum_{j=1}^{n}|j\rangle_1\frac{1}{\sqrt{p+q}}\bigg(\sum_{i=1}^{p}|i\rangle_2|\bar{M}_{i,*}\rangle_3|\acute{E}_{ij}\rangle_4|00\rangle_5\\
&+\sum_{i=p+1}^{p+q}|i\rangle_2|\bar{M}_{i,*}\rangle_3|0\rangle_4|01\rangle_5\bigg)\\
&+\sum_{j=n+1}^{2n}|j\rangle_1\frac{1}{\sqrt{p+q}}\bigg(\sum_{i=1}^{p}|i\rangle_2|\bar{M}_{i,*}\rangle_3|0\rangle_4|10\rangle_5\\
&+\sum_{i=p+1}^{p+q}|i\rangle_2|\bar{M}_{i,*}\rangle_3|\acute{E}_{ij}\rangle_4|11\rangle_5\bigg)\bigg],
\end{aligned}
\end{eqnarray}
where $U_{QMA}$ represents QMA acting on the third and fourth registers.

(1.6) Add an ancillary qubit $|0\rangle$, carry out a appropriate controlled rotation on the ancillary qubit, then uncompute the third, fourth and fifth registers to get
\begin{eqnarray}
\begin{aligned}
&\frac{1}{\sqrt{2n}}\sum_{j=1}^{2n}|j\rangle_1\frac{1}{\sqrt{p+q}}\sum_{i=1}^{p+q}|i\rangle_2\bigg(\frac{\acute{E}_{ij}}{\alpha}|0\rangle_6\\
&+\sqrt{1-(\frac{\acute{E}_{ij}}{\alpha})^2}|1\rangle_6\bigg),
\end{aligned}
\end{eqnarray}
where $\alpha=2\max_{ij}|M_{ij}|$. See Appendix \ref{app:pa} for more details of $\alpha$.

(1.7) Amplifying the amplitude of $|0\rangle_6$ by fixed-point quantum search~\cite{fix2005,fix2014}, then we can get the approximate state of $|\psi_E\rangle$.

(2) Prepare the state $|\psi_J\rangle$.

Note that we can rewrite $\mathcal{M}$ as $\mathcal{M}=[\mathcal{M}^1,\mathcal{M}^2,...,\mathcal{M}^c]$, where $\mathcal{M}^i=\left(\begin{matrix}
   \mathcal{A}^i   \\
   \mathcal{B}^i  \\
  \end{matrix}
  \right)$, $i=1,2,...,c$. From the construction of matrices $\mathcal{M}$ and $\acute{J}$, we know $\acute{J}_{ki} =n'\bar{\mathcal{M}}^i_{k,*}-n_i\bar{M}_{k,*}$ where  $\bar{\mathcal{M}}^i_{k,*}$ is the mean of elements in the $k$th row of $\mathcal{M}^i$. We can use \emph{Lemma}~\ref{lem:ME} to calculate $\bar{\mathcal{M}}^i_{k,*}$ and $\bar{M}_{i,*}$, then use $\textbf{O}_c$ and QMA to obtain the state $\frac{1}{\sqrt{c}}\sum_{i=1}^{c}|i\rangle\frac{1}{\sqrt{p+q}}\sum_{k=1}^{p+q}|k\rangle|\acute{J}_{ki}\rangle$. Afterwards, we use controlled rotation and fixed-point quantum search to get $|\psi_J\rangle$. Details are as follows.

(2.1) Prepare the initial state
\begin{eqnarray}
\begin{aligned}
\frac{1}{\sqrt{c}}\sum_{i=1}^{c}|i\rangle_1\frac{1}{\sqrt{p+q}}\sum_{k=1}^{p+q}|k\rangle_2|0\rangle_3|0\rangle_4|0\rangle_5|0\rangle_6.
\end{aligned}
\end{eqnarray}

(2.2) By \emph{Lemma}~\ref{lem:ME} and a unitary $U_{\mathcal{M}}: |i\rangle|k\rangle|j\rangle|0\rangle\rightarrow |i\rangle|k\rangle|j\rangle|\mathcal{M}_{kj}^{i}\rangle$ for $j=1,2,...,n'$, we can realize the unitary $\widetilde{U}_{mean}: |i\rangle|k\rangle|0\rangle\rightarrow|i\rangle|k\rangle|\bar{\mathcal{M}}^i_{k,*}\rangle$. Then the following state can be obtained after performing $\widetilde{U}_{mean}$ on the first three registers.
\begin{eqnarray}
\begin{aligned}
&\frac{1}{\sqrt{c}}\sum_{i=1}^{c}|i\rangle_1\frac{1}{\sqrt{p+q}}\sum_{k=1}^{p+q}|k\rangle_2|\bar{\mathcal{M}}^i_{k,*}\rangle_3|0\rangle_4|0\rangle_5|0\rangle_6.
\end{aligned}
\end{eqnarray}

The details of $U_{\mathcal{M}}$ are as follows. According to the construction of matrix $\mathcal{M}$, we have $\mathcal{M}_{kj}^{i}=\mathcal{M}_{k,(i-1)n'+j}$ for $j=1,2,...,n'$. We can first realize the unitary mapping: $|i\rangle|k\rangle|j\rangle|0\rangle|0\rangle
\rightarrow |i\rangle|k\rangle|j\rangle|(i-1)n'+j\rangle|0\rangle$, then perform $\textbf{O}_{\mathcal{M}}$ on the second, fourth and fifth registers. The state $|i\rangle|k\rangle|j\rangle|\mathcal{M}_{kj}^{i}\rangle$ can be obtained after discarding the fourth register.

(2.3) Perform $U_f : |x\rangle|0\rangle\rightarrow |x\rangle|f(x)\rangle$ on the third and fourth registers with function $f(x)=n'x$ which can be calculated efficiently in classical. Then we can get
\begin{eqnarray}
\begin{aligned}
&\frac{1}{\sqrt{c}}\sum_{i=1}^{c}|i\rangle_1\frac{1}{\sqrt{p+q}}\sum_{k=1}^{p+q}|k\rangle_2|x\rangle_3|n'x\rangle_4|0\rangle_{5}|0\rangle_6,
\end{aligned}
\end{eqnarray}
where $x=\bar{\mathcal{M}}^i_{k,*}$.

(2.4) By \emph{Lemma}~\ref{lem:ME}, with a given $\textbf{O}_M$, we can perform $U_{mean}$ on the second and fifth registers to yield
\begin{eqnarray}
\begin{aligned}
&\frac{1}{\sqrt{c}}\sum_{i=1}^{c}|i\rangle_1\frac{1}{\sqrt{p+q}}\sum_{k=1}^{p+q}|k\rangle_2|x\rangle_3|n'x\rangle_4
|\bar{M}_{k,*}\rangle_5|0\rangle_6.
\end{aligned}
\end{eqnarray}

(2.5) Apply $\textbf{O}_c$ to the first and sixth registers. The following state can be obtained after implementing QMA on the fourth, fifth and sixth registers.
\begin{eqnarray}
\begin{aligned}%{}里面负数为左移，正为右移
&\frac{1}{\sqrt{c}}\sum_{i=1}^{c}|i\rangle_1\frac{1}{\sqrt{p+q}}\sum_{k=1}^{p+q}|k\rangle_2|x\rangle_3
|\acute{J}_{ki}\rangle_4|\bar{M}_{k,*}\rangle_5|n_i\rangle_6,
\end{aligned}
\end{eqnarray}
where $\acute{J}_{ki}=n'\bar{\mathcal{M}}^i_{k,*}-n_i\bar{M}_{k,*}$.

(2.6) Add a qubit and rotating conditioned on $|\acute{J}_{ki}\rangle_4$ to get
\begin{eqnarray}
\begin{aligned}
\hspace{-7mm}%{}里面负数为左移，正为右移
&\frac{1}{\sqrt{c}}\sum_{i=1}^{c}|i\rangle_1\frac{1}{\sqrt{p+q}}\sum_{k=1}^{p+q}|k\rangle_2|x\rangle_3|\acute{J}_{ki}\rangle_4
\\&|\bar{M}_{k,*}\rangle_5|n_i\rangle_6\bigg(\frac{\acute{J}_{ki}}{\beta}|0\rangle_7
+\sqrt{1-(\frac{\acute{J}_{ki}}{\beta})^2}|1\rangle_7\bigg).
\end{aligned}
\end{eqnarray}
where $\beta=2n'\max_{ij}|M_{ij}|$. See Appendix \ref{app:pa} for more details of $\beta$.

(2.7) Uncompute the redundant registers. Then we can use fixed-point quantum search to amplify the amplitude of $|0\rangle_7$ to get the approximate state of $|\psi_J\rangle$.

(3) Prepare the state $|\psi_K\rangle$.

From the construction of matrices $\mathcal{M}$ and $\acute{K}$, we have
\begin{eqnarray}
    \begin{aligned}
    \acute{K}_{ki}=
            \begin{cases}
            n'\bar{\mathcal{M}}^i_{k,*}-n_i\bar{M}_{k,*}, & 1\leq k\leq p,\ 1\leq i\leq c;\\
n'\bar{\mathcal{M}}^{i-c}_{k,*}-n_{i-c}\bar{M}_{k,*}, & p+1\leq k\leq p+q, \\
& c+1\leq i\leq 2c; \\
0, & \text{otherwise}.\\
            \end{cases}
    \end{aligned}
    \end{eqnarray}
Similar to $|\psi_J\rangle$, we can first prepare the state $\frac{1}{\sqrt{2c}}\sum_{i=1}^{2c}|i\rangle\frac{1}{\sqrt{p+q}}\sum_{k=1}^{p+q}|k\rangle|\acute{K}_{ki}\rangle$, then use controlled rotation and fixed-point quantum search to get $|\psi_K\rangle$.

(3.1) Prepare the initial state
\begin{eqnarray}
\begin{aligned}
\frac{1}{\sqrt{2c}}\sum_{i=1}^{2c}|i\rangle_1\frac{1}{\sqrt{p+q}}\sum_{k=1}^{p+q}|k\rangle_2|0\rangle_3|0\rangle_4|0\rangle_5|0\rangle_6|0\rangle_7^{\otimes 2}.\ \
\end{aligned}
\end{eqnarray}

(3.2) Perform $U=\sum_{i=1}^c|i\rangle\langle i|\otimes \sum_{k=1}^p|k\rangle\langle k|\otimes (I\otimes I) + \sum_{i=1}^c|i\rangle\langle i|\otimes \sum_{k=p+1}^{p+q}|k\rangle\langle k|\otimes (I\otimes X)+\sum_{i=c+1}^{2c}|i\rangle\langle i|\otimes \sum_{k=1}^{p}|k\rangle\langle k|\otimes (X\otimes I)+\sum_{i=c+1}^{2c}|i\rangle\langle i|\otimes \sum_{k=p+1}^{p+q}|k\rangle\langle k|\otimes (X\otimes X)$ on the first, second and seventh registers to get
\begin{eqnarray}
\begin{aligned}
&\frac{1}{\sqrt{2c}}\bigg[\sum_{i=1}^{c}|i\rangle_1\frac{1}{\sqrt{p+q}}\bigg(\sum_{k=1}^{p}|k\rangle_2|0\rangle_3|0\rangle_4|0\rangle_5|0\rangle_6|00\rangle_7
\\
&+\sum_{k=p+1}^{p+q}|k\rangle_2|0\rangle_3|0\rangle_4|0\rangle_5|0\rangle_6|01\rangle_7\bigg)\\
&+\sum_{i=c+1}^{2c}|i\rangle_1\frac{1}{\sqrt{p+q}}\bigg(\sum_{k=1}^{p}|k\rangle_2|0\rangle_3|0\rangle_4|0\rangle_5|0\rangle_6|10\rangle_7\\
&+\sum_{k=p+1}^{p+q}|k\rangle_2|0\rangle_3|0\rangle_4|0\rangle_5|0\rangle_6|11\rangle_7\bigg)\bigg],
\end{aligned}
\end{eqnarray}
where $X$ is the Pauli-$X$ gate and $I$ is the identity operator.

(3.3) Given a unitary $\tilde{U}_{\mathcal{M}}$ which can be used to perform the mapping $|i\rangle|k\rangle|j\rangle|0\rangle\rightarrow |i\rangle|k\rangle|j\rangle|\mathcal{M}_{kj}^{i-c}\rangle$ for $i=c+1,c+2,...,2c$ and $j=1,2,...,n'$, we can realize the unitary $U:|i\rangle|k\rangle|0\rangle\rightarrow |i\rangle|k\rangle|\bar{\mathcal{M}}^{i-c}_{k,*}\rangle$, $i=c+1,c+2,...,2c$ by \emph{Lemma}~\ref{lem:ME}. The realization of unitary $\tilde{U}_{\mathcal{M}}$ is similar to $U_{\mathcal{M}}$ in stage (2.2). With the $\widetilde{U}_{mean}$ in stage (2.2), we perform $\widetilde{U}_{mean}\otimes |00\rangle \langle 00|+U\otimes |11\rangle \langle 11|+I_{1,2,3}\otimes(|01\rangle \langle 01|+|10\rangle \langle 10|)$ on the first three registers and the seventh register to get
\begin{eqnarray}
\begin{aligned}
&\frac{1}{\sqrt{2c}}\bigg[\sum_{i=1}^{c}|i\rangle_1\frac{1}{\sqrt{p+q}}\bigg(\sum_{k=1}^{p}|k\rangle_2|\bar{\mathcal{M}}^i_{k,*}\rangle_3|0\rangle_4|0\rangle_5|0\rangle_6|00\rangle_7
\\
&+\sum_{k=p+1}^{p+q}|k\rangle_2|0\rangle_3|0\rangle_4|0\rangle_5|0\rangle_6|01\rangle_7\bigg)\\
&+\sum_{i=c+1}^{2c}|i\rangle_1\frac{1}{\sqrt{p+q}}\bigg(\sum_{k=1}^{p}|k\rangle_2|0\rangle_3|0\rangle_4|0\rangle_5|0\rangle_6|10\rangle_7\\
&+\sum_{k=p+1}^{p+q}|k\rangle_2|\bar{\mathcal{M}}^{i-c}_{k,*}\rangle_3|0\rangle_4|0\rangle_5|0\rangle_6|11\rangle_7\bigg)\bigg].
\end{aligned}
\end{eqnarray}

(3.4) Similar to stage (2.3), the unitary $U_f : |x\rangle|0\rangle\rightarrow |x\rangle|f(x)\rangle$ with function $f(x)=n'x$ can be performed efficiently. We implement $U_f\otimes (|00\rangle \langle 00|+|11\rangle \langle 11|)+I_{3,4}\otimes(|01\rangle \langle 01|+|10\rangle \langle 10|)$ on the third, fourth and seventh registers to obtain
\begin{eqnarray}
\begin{aligned}
&\frac{1}{\sqrt{2c}}\bigg[\sum_{i=1}^{c}|i\rangle_1\frac{1}{\sqrt{p+q}}\bigg(\sum_{k=1}^{p}|k\rangle_2|\bar{\mathcal{M}}^i_{k,*}\rangle_3|n'\bar{\mathcal{M}}^i_{k,*}\rangle_4
\\
&|0\rangle_5|0\rangle_6|00\rangle_7
+\sum_{k=p+1}^{p+q}|k\rangle_2|0\rangle_3|0\rangle_4|0\rangle_5|0\rangle_6|01\rangle_7\bigg)\\
&+\sum_{i=c+1}^{2c}|i\rangle_1\frac{1}{\sqrt{p+q}}\bigg(\sum_{k=1}^{p}|k\rangle_2|0\rangle_3|0\rangle_4|0\rangle_5|0\rangle_6|10\rangle_7\\
&+\sum_{k=p+1}^{p+q}|k\rangle_2|\bar{\mathcal{M}}^{i-c}_{k,*}\rangle_3|n'\mathcal{M}^{i-c}_{k,*}\rangle_4|0\rangle_5|0\rangle_6|11\rangle_7\bigg)\bigg],
\end{aligned}
\end{eqnarray}
where $x=\bar{\mathcal{M}}^i_{k,*}$ or $\bar{\mathcal{M}}^{i-c}_{k,*}$.

(3.5) By \emph{Lemma}~\ref{lem:ME}, with a given $\textbf{O}_M$, we can perform $U_{mean}$ on the second and fifth registers to yield
\begin{eqnarray}
\begin{aligned}
&\frac{1}{\sqrt{2c}}\bigg[\sum_{i=1}^{c}|i\rangle_1\frac{1}{\sqrt{p+q}}\bigg(\sum_{k=1}^{p}|k\rangle_2|\bar{\mathcal{M}}^i_{k,*}\rangle_3|n'\bar{\mathcal{M}}^i_{k,*}\rangle_4|\bar{M}_{k,*}\rangle_5
\\
&|0\rangle_6|00\rangle_7
+\sum_{k=p+1}^{p+q}|k\rangle_2|0\rangle_3|0\rangle_4|\bar{M}_{k,*}\rangle_5|0\rangle_6|01\rangle_7\bigg)\\
&+\sum_{i=c+1}^{2c}|i\rangle_1\frac{1}{\sqrt{p+q}}\bigg(\sum_{k=1}^{p}|k\rangle_2|0\rangle_3|0\rangle_4|\bar{M}_{k,*}\rangle_5|0\rangle_6|10\rangle_7\\
&+\sum_{k=p+1}^{p+q}|k\rangle_2|\bar{\mathcal{M}}^{i-c}_{k,*}\rangle_3|n'\mathcal{M}^{i-c}_{k,*}\rangle_4|\bar{M}_{k,*}\rangle_5|0\rangle_6|11\rangle_7\bigg)\bigg].
\end{aligned}
\end{eqnarray}

(3.6) Using $\textbf{O}_c$, we can realize a unitary $U_c: |i\rangle|0\rangle\rightarrow |i\rangle|n_{i-c}\rangle$ for $i=c+1,c+2,...,2c$. We then implement $\textbf{O}_c\otimes|00\rangle \langle 00| +U_c\otimes|11\rangle \langle 11|+I_{1,6}\otimes(|01\rangle \langle 01|+|10\rangle \langle 10|)$ on the first, sixth and seventh registers to get
\begin{eqnarray}
\begin{aligned}
&\frac{1}{\sqrt{2c}}\bigg[\sum_{i=1}^{c}|i\rangle_1\frac{1}{\sqrt{p+q}}\bigg(\sum_{k=1}^{p}|k\rangle_2|\bar{\mathcal{M}}^i_{k,*}\rangle_3|n'\bar{\mathcal{M}}^i_{k,*}\rangle_4|\bar{M}_{k,*}\rangle_5
\\
&|n_i\rangle_6|00\rangle_7
+\sum_{k=p+1}^{p+q}|k\rangle_2|0\rangle_3|0\rangle_4|\bar{M}_{k,*}\rangle_5|0\rangle_6|01\rangle_7\bigg)\\
&+\sum_{i=c+1}^{2c}|i\rangle_1\frac{1}{\sqrt{p+q}}\bigg(\sum_{k=1}^{p}|k\rangle_2|0\rangle_3|0\rangle_4|\bar{M}_{k,*}\rangle_5|0\rangle_6|10\rangle_7\\
&+\sum_{k=p+1}^{p+q}|k\rangle_2|\bar{\mathcal{M}}^{i-c}_{k,*}\rangle_3|n'\mathcal{M}^{i-c}_{k,*}\rangle_4|\bar{M}_{k,*}\rangle_5|n_{i-c}\rangle_6|11\rangle_7\bigg)\bigg].
\end{aligned}
\end{eqnarray}

(3.7) Perform $U_{QMA}\otimes(|00\rangle \langle 00| +\otimes|11\rangle \langle 11|)+I_{4,5,6}\otimes(|01\rangle \langle 01|+|10\rangle \langle 10|)$ on the fourth, fifth, sixth and seventh registers, then uncompute the redundant registers to get
\begin{eqnarray}
\begin{aligned}
\frac{1}{\sqrt{2c}}\sum_{i=1}^{2c}|i\rangle_1\frac{1}{\sqrt{p+q}}\sum_{k=1}^{p+q}|k\rangle_2|\acute{K}_{ki}\rangle_4,
\end{aligned}
\end{eqnarray}
where $U_{QMA}$ represents QMA acting on the fourth, fifth and sixth registers.

(3.8) Similar to stages (2.6)-(2.7), $|\psi_K\rangle$ can be obtained by implementing controlled rotation and fixed-point quantum search.

Note that we also complete the mean-centering of DCCA while we prepare the three states.

\emph{\textbf{Step 2.} Design the block-encoding of $\tilde{H}$. }

The block-encodings of $\rho_E$, $\rho_J$ and $\rho_K$ can be obtained easily according to Lemma 25 in Ref.~\cite{GSL}.

We now show how to build up the block-encoding of $\tilde{H}$. We first construct the block-encodings of $\rho_E^{-\frac{1}{2}}$ according to Lemma 9 in Ref.~\cite{CGJ}, and next realize the block-encodings of $\rho_E^{-\frac{1}{2}}\rho_J\rho_E^{-\frac{1}{2}}$ as well as $\rho_E^{-\frac{1}{2}}\rho_K\rho_E^{-\frac{1}{2}}$ by product of block-encoded matrices~\cite{GSL}. It is obvious that $(P_L,P_R)$ is a (2,1,0)-state-preparation-pair when $P_L=HX$ and $P_R=H$, where $H$ represents a Hadamard gate and $X$ is a Pauli-$X$ gate. Then, the block-encoding of $\tilde{H}$ can be created according to linear combination of block-encoded matrices~\cite{CGJ1}. We summarize the construction parameters of block-encodings as TABLE~\ref{tab:be}.

\begin{table*}[!htb]
\caption{\label{tab:be}
The construction of block-encodings.}
\begin{ruledtabular}
\begin{tabular}{ccc}
matrices  &block-encodings of matrices &parameters of block-encodings \\ \hline
$\rho_E$&$U_{\rho_E}=((U_E)^\dag\otimes I_s)(I_{a_E}\otimes SWAP_s)(U_E\otimes I_s)$ (Lemma 25 in Ref.~\cite{GSL})&$(1,a_E+s,2\epsilon_E)$  \\
$\rho_J$&$U_{\rho_J}=((U_J)^\dag\otimes I_s)(I_{a_J}\otimes SWAP_s)(U_J\otimes I_s)$ (Lemma 25 in Ref.~\cite{GSL})&$(1,a_J+s,2\epsilon_J)$\\
$\rho_K$&$U_{\rho_K}=((U_K)^\dag\otimes I_s)(I_{a_K}\otimes SWAP_s)(U_K\otimes I_s)$ (Lemma 25 in Ref.~\cite{GSL})&$(1,a_K+s,2\epsilon_K)$  \\
$\rho_E^{-\frac{1}{2}}$&$\tilde{U}_{\rho_E}$ see Lemma 9 in Ref.~\cite{CGJ}& $(2\kappa^{\frac{1}{2}},a',\epsilon_3)$\\
$F :=\rho_E^{-\frac{1}{2}}\rho_J\rho_E^{-\frac{1}{2}}$  &$U_{F}$ see Lemma 30 in Ref.~\cite{GSL}&$(4\kappa,a'',4\kappa^{\frac{1}{2}}(\epsilon_3+2\kappa^{\frac{1}{2}}\epsilon_J)$ \\
$G :=\rho_E^{-\frac{1}{2}}\rho_K\rho_E^{-\frac{1}{2}}$       &$U_{G}$ see Lemma 30 in Ref.~\cite{GSL} &$(4\kappa,a''',4\kappa^{\frac{1}{2}}(\epsilon_3+2\kappa^{\frac{1}{2}}\epsilon_K)$ \\
$\tilde{H}$       &$U_{\tilde{H}}$ see Lemma 52 in Ref.~\cite{CGJ1}&$(8\kappa,a'''+1,32\kappa^{\frac{3}{2}}(\epsilon_3+2\kappa^{\frac{1}{2}}\epsilon_J))$ \\
%$e^{i\tilde{H}t}$       &$U_{e^{i\tilde{H}t}}$, refer to theorem 7 in ~~\cite{CGJ}&$(1,a'''+3,|2t|\cdot32\kappa^{\frac{3}{2}}(\epsilon_3+2\kappa^{\frac{1}{2}}\epsilon_J))$ \\
\end{tabular}
\end{ruledtabular}
Here $U_E$, $U_J$, $U_K$ denote the unitary operations of preparing the states $|\psi_E\rangle$, $|\psi_J\rangle$, $|\psi_K\rangle$ respectively, and $\epsilon_E$, $\epsilon_J$, $\epsilon_K$ are their corresponding errors. $SWAP_s$ denotes a SWAP gate between the second register and an ancillary system, and $I_s$ is the identity operator acting on $s$ qubits. $s=\log(p+q)$, $a'=a_E+s+O(\log(\kappa^{\frac{3}{2}}\log\frac{1}{\epsilon_3}))$, $a''=a_J+s+2a'$, $a'''=a_K+s+2a'$, $\epsilon_E=O(\frac{\epsilon_3}{\kappa^{3/2}\log^3 (\frac{\kappa^{3/2}}{\epsilon_3})})$, $\kappa$ is the condition number of $\rho_E$. $a_E=\log n+\log m_1+4$, $a_J=\log c+\log m_2+1$, $a_K=a_J+3$, $\log m_1$ is the number of qubits in the third and fourth registers when preparing $|\psi_E\rangle$, $\log m_2$ is the number of qubits in the third, fourth, fifth and sixth registers when preparing $|\psi_J\rangle$.
\end{table*}

\emph{\textbf{Step 3.} Estimate the eigenvalues of $H$ according to quantum phase estimation, and then search the first $d$ largest eigenvalues of $H$ to get the corresponding eigenvectors $\{|\textbf{v}_i\rangle\}_{i=1}^d$. }

Given the block-encoding of $\tilde{H}$, the unitary $e^{i\tilde{H}t}$ can be implemented according to block-Hamiltonian simulation (Theorem 3 in Ref.~\cite{ICALP2019}). By using $e^{i\tilde{H}t}$, we apply quantum phase estimation on $\rho_0:=\frac{1}{p+q}\sum_{k=1}^{p+q}|k\rangle\langle k|$ to obtain an approximation to the state
\begin{eqnarray}
\begin{aligned}
\rho_1=\frac{1}{p+q}\sum_{k=1}^{p+q}|\lambda_k\rangle\langle \lambda_k|\otimes|\textbf{v}_k\rangle \langle \textbf{v}_k|,
\end{aligned}
\end{eqnarray}
where $\lambda_k$ and $\textbf{v}_k$ are the eigenvalues and eigenvectors of $H$. The state $\rho_0$ can be prepared easily by Hadamard and CNOT gates.

Afterwards, we invoke the quantum search algorithm for finding the maximum~\cite{ahuja1999} to find the first $d$ largest eigenvalues of $H$ and the corresponding eigenvectors $\{|\textbf{v}_i\rangle\}_{i=1}^d$.

\emph{\textbf{Step 4.} Postprocessing}.

According to TABLE~\ref{tab:be}, we first create a $(1,a_E+s,2\epsilon_E)$-block-encoding of $\rho_E$. Then, for each $|\textbf{v}_i\rangle$, we use quantum matrix inversion technique (Theorem 10 in Ref.~\cite{ICALP2019}) to get the state $|\textbf{w}_{i}\rangle :=\frac{(\rho_E)^{-1/2}|\textbf{v}_i\rangle}{\|(\rho_E)^{-1/2}|\textbf{v}_i\rangle\|_2}\propto \textbf{w}_{i}$, where $i=1,2,...,d$.

\subsection{Complexity analysis}\label{subsec:CA}

In this section we analyze the time complexity of each step of the QDCCA algorithm and summarize it as TABLE~\ref{tab:stc}.

\begin{table*}[!htb]
\caption{\label{tab:stc}
The time complexity of each step of the QDCCA algorithm.}
\begin{ruledtabular}
\begin{tabular}{ccc}
steps  &unitary operations&time complexity \\ \hline
\emph{\textbf{Step 1}} &$U_E$&$T_E=O(\frac{(\max_{ij}|M_{ij}|)^2\log(n(p+q))\log(1/\Delta_1)}{m_0\epsilon_1} )$ \\
  &$U_J$&$T_J=O(\frac{(\max_{ij}|M_{ij}|)^3\log(n(p+q))\log(cn'(p+q))\log(1/\Delta_1)\log(1/\Delta_2)}{m_0\epsilon_1\epsilon_2} )$\\
 &$U_K$&$T_K=T_J$ \\ \hline
\emph{\textbf{Step 2}} &$U_{\rho_E}$, $U_{\rho_J}$, $U_{\rho_K}$&$T_E$, $T_J$ and $T_K$ respectively \\
&$\tilde{U}_{\rho_E}$& $\tilde{T}_E=O(\kappa(a_E+s+T_E)\log^2(\frac{\kappa^{3/2}}{\epsilon_3}))$\\
        &$U_{F}$&$\tilde{T}_E+T_J$ \\
       &$U_{G}$&$\tilde{T}_E+T_K$ \\
       &$U_{\tilde{H}}$&$\tilde{T}_E+T_J+T_K$ \\
  %     &$U_{e^{i\tilde{H}t}}$&$O((|8\kappa t|+\log\frac{1}{|2t|\cdot \epsilon_{\tilde{H}}})(\tilde{T}_E+T_J+T_K))$ \\
        \hline
\emph{\textbf{Step 3}}  &quantum phase estimation, quantum search& $O(d\sqrt{p+q}(\frac{n(\max_{ij}|M_{ij}|)^2\kappa}{m_0^2\epsilon_4}+\log\frac{m_0^2\epsilon_4}{n(\max_{ij}|M_{ij}|)^2\epsilon_{\tilde{H}}})(\tilde{T}_E+T_J+T_K))$ \\        \hline
%\emph{\textbf{Step 4}} &quantum search& $O(d\sqrt{p+q})$ \\\hline
\emph{\textbf{Step 4}} &quantum matrix inversion technique & $O(\kappa T_{1-3}\log \kappa)$  \\\hline
all steps &---& $O(\frac{d\sqrt{p+q} n(\max_{ij}|M_{ij}|)^5\kappa^3 \log^2(\frac{\kappa^{3/2}}{\epsilon_3})\log^2(n(p+q))\log \kappa}{m_0^3\epsilon_1\epsilon_2\epsilon_4})$ \\
\end{tabular}
\end{ruledtabular}
Here we follow the notations in TABLE~\ref{tab:be}. In addition, $\epsilon_1$ is the error of $\bar{M}_{i,*}$, $\epsilon_2$ is the error of $\bar{\mathcal{M}}^i_{k,*}$, $\epsilon_{\tilde{H}}=32\kappa^{\frac{3}{2}}(\epsilon_3+2\kappa^{\frac{1}{2}}\epsilon_J))$, $\epsilon_4$ is the error of quantum phase estimation, $T_{1-3}$ is the total time complexity from \emph{\textbf{Step 1}} to \emph{\textbf{Step 3}}. For simplicity, the factors $\log(1/\Delta_1)$ and $\log(1/\Delta_2)$ can be consider as constants, and $a_E+s$ in $\tilde{T}_E$ can be ignored.
%and we use the notation $\tilde{O}(f(x))$ to indicate $O(f(x)\mathrm {polylog}(f(x))$.
\end{table*}

In \emph{\textbf{Step 1}}, for stage (1), since the complexity of Pauli-$X$ gates, $U_{QMA}$ and controlled rotation are much smaller than the complexity of other stages, we will neglect the complexity of stages (1.3), (1.5) and (1.6). The complexity of stage (1.1) is $O[\log (n(p+q))]$. For stage (1.2), by \emph{Lemma}~\ref{lem:ME}, we use the unitary $U_{mean}$ with complexity $O(\frac{\max_{ij}|M_{ij}|\log(n(p+q))\log(1/\Delta_1)}{\epsilon_1} )$ to get the target state with a probability at least $1-2\Delta_1$ where $\epsilon_1$ is the error of $\bar{M}_{i,*}$. The complexity of stage (1.4) comes mainly from $\textbf{O}_{M}$ and $U_M$, and the number of gates required of $U_M$ is roughly equal to $\textbf{O}_{M}$. Hence, the complexity of stage (1.4) is $O[\log(n(p+q))]$. For stage (1.7), we assume that the proportion of elements in $\left(\begin{matrix}
   X  \\
   Y\\
  \end{matrix}
  \right)$ with absolute value greater than $m_0>0$ is at least $\frac{1}{2}$ (or other reasonable constants), to obtain the state $|\psi_E\rangle$ with a probability close to 1, the complexity of the fixed-point quantum search is
\begin{eqnarray*}
\begin{aligned}
        \sqrt{\frac{2n(p+q)\alpha^2}{\|\acute{E}\|_F^2}}
&\leq\sqrt{\frac{2n(p+q)\alpha^2}{(1/2)n(p+q)(m_0)^2}}
\\&=O(\frac{\max_{ij}|M_{ij}|}{m_0}).
\end{aligned}
\end{eqnarray*}
Therefore, the complexity of stage (1) is $O(\frac{(\max_{ij}|M_{ij}|)^2\log(n(p+q))\log(1/\Delta_1)}{m_0\epsilon_1} )$.

  For stage (2) of \emph{\textbf{Step 1}}, the complexity of (2.3) and (2.6) can be neglected. Due to the use of $U_{mean}$ in (2.4), we should prepare the target state of stage (2.3) repeatedly. The complexity of stage (2.1) is $O[\log (c(p+q))]$. By \emph{Lemma}~\ref{lem:ME}, we get the target state of stage (2.2) with a probability at least $1-2\Delta_2$ in time $O(\frac{\max_{ij}|M_{ij}|\log(cn'(p+q))\log(1/\Delta_2)}{\epsilon_2} )$ where $\epsilon_2$ is the error of $\bar{\mathcal{M}}^i_{k,*}$, then the complexity of stage (2.4) is $O(\frac{(\max_{ij}|M_{ij}|)^2\log(n(p+q))\log(cn'(p+q))\log(1/\Delta_1)\log(1/\Delta_2)}{\epsilon_1\epsilon_2} )$.
  For stage (2.5), the complexity of $\textbf{O}_c$ is $O(\log c)$ and can be omitted. For stage (2.7), due to the fact that the characteristics of samples within-class are similar, we can assume that the proportion of elements in $\acute{J}$ with absolute value greater than $n''m_0$ is at least $\frac{1}{2}$ (or other reasonable constants) where $n''=\min(n_i)$, then the complexity of the fixed-point quantum search is
\begin{eqnarray*}
\begin{aligned}
\sqrt{\frac{c(p+q)\beta^2}{\|\acute{J}\|_F^2}}
       &\leq\sqrt{\frac{c(p+q)\beta^2}{(1/2)c(p+q)(n''m_0)^2}}
         \\&=O(\frac{\max_{ij}|M_{ij}|}{m_0}).
\end{aligned}
\end{eqnarray*}
Therefore, the complexity of stage (2) is $O(\frac{(\max_{ij}|M_{ij}|)^3\log(n(p+q))\log(cn'(p+q))\log(1/\Delta_1)\log(1/\Delta_2)}{m_0\epsilon_1\epsilon_2} )$.

For stage (3) of \emph{\textbf{Step 1}}, the complexity of stages (3.2), (3.6) and (3.7) can be neglected. The complexity of stage (3.1) is $O[\log (2c(p+q))]$. Since the complexity of $\tilde{U}_{\mathcal{M}}$ is same to $U$, the complexity of stage (3.3) is $O(\frac{\max_{ij}|M_{ij}|\log(cn'(p+q))\log(1/\Delta_2)}{\epsilon_2} )$ by \emph{Lemma}~\ref{lem:ME}.
For stage (3.5), we should prepare the target state of stage (3.4) repeatedly. Because the complexity of stages (3.1)-(3.4) is $O(\frac{\max_{ij}|M_{ij}|\log(cn'(p+q))\log(1/\Delta_2)}{\epsilon_2} )$, the complexity of stage (3.5) is $O(\frac{(\max_{ij}|M_{ij}|)^2\log(n(p+q))\log(cn'(p+q))\log(1/\Delta_1)\log(1/\Delta_2)}{\epsilon_1\epsilon_2} )$. The complexity of stage (3.8) is $O(\frac{\max_{ij}|M_{ij}|}{m_0})$. In summary, the complexity of stage (3) is equal to stage (2).

Let $\widetilde{|\psi_E\rangle}$ represent the approximate state of $|\psi_E\rangle$ which we prepared and $\epsilon_{ij}$ is the error of $\acute{E}_{ij}$. Then, the error of $|\psi_E\rangle$ is
\begin{small}
\begin{eqnarray*}
\begin{aligned}
\epsilon_E&=\bigg\|\widetilde{|\psi_E}\rangle-|\psi_E\rangle\bigg\|_2
\\&=\bigg\|\frac{1}{\|\widetilde{\acute{E}}\|_F}\sum_{ij}(\acute{E}_{ij}+\epsilon_{ij})|j\rangle|i\rangle-\frac{1}{\|\acute{E}\|_F}\sum_{ij}\acute{E}_{ij}|j\rangle|i\rangle\bigg\|_2
\\&\leq\bigg\|\frac{1}{\|\widetilde{\acute{E}}\|_F}\sum_{ij}(\acute{E}_{ij}+\epsilon_{ij})|j\rangle|i\rangle
-\frac{1}{\|\acute{E}\|_F}\sum_{ij}(\acute{E}_{ij}+\epsilon_{ij})|j\rangle|i\rangle\bigg\|_2
\\&+\bigg\|\frac{1}{\|\acute{E}\|_F}\sum_{ij}(\acute{E}_{ij}+\epsilon_{ij})|j\rangle|i\rangle
-\frac{1}{\|\acute{E}\|_F}\sum_{ij}\acute{E}_{ij}|j\rangle|i\rangle\bigg\|_2
\end{aligned}
\end{eqnarray*}
\end{small}

\begin{small}
\begin{eqnarray*}
\begin{aligned}
&=\bigg|1-\frac{\|\widetilde{\acute{E}}\|_F}{\|\acute{E}\|_F}\bigg|+\bigg\|\frac{1}{\|\acute{E}\|_F}\sum_{ij}\epsilon_{ij}|j\rangle|i\rangle\bigg\|_2
\\&=\bigg|1-\sqrt{\frac{\sum_{ij}(\acute{E}_{ij}+\epsilon_{ij})^2}{\sum_{ij}(\acute{E}_{ij})^2}}\bigg|+\sqrt{\frac{\sum_{ij}(\epsilon_{ij})^2}{\sum_{ij}(\acute{E}_{ij})^2}},
\end{aligned}
\end{eqnarray*}
\end{small}
where $\|\widetilde{\acute{E}}\|_F=\sqrt{\sum_{ij}(\acute{E}_{ij}+\epsilon_{ij})^2}$. For simplicity, we assume that $1-\sqrt{\frac{\sum_{ij}(\acute{E}_{ij}+\epsilon_{ij})^2}{\sum_{ij}(\acute{E}_{ij})^2}}\leq0$. A similar result can be obtained if it is greater than 0. Let $\epsilon_{ij}=\epsilon_1$, then
\begin{small}
\begin{eqnarray*}
\begin{aligned}
\epsilon_E&\leq\sqrt{1+\frac{2\sum_{ij}\acute{E}_{ij}\epsilon_{ij}}{\sum_{ij}(\acute{E}_{ij})^2}+\frac{\sum_{ij}(\epsilon_{ij})^2}{\sum_{ij}(\acute{E}_{ij})^2}}-1+\sqrt{\frac{\sum_{ij}(\epsilon_{ij})^2}{\sum_{ij}(\acute{E}_{ij})^2}}
\end{aligned}
\end{eqnarray*}
\end{small}

\begin{small}
\begin{eqnarray*}
\begin{aligned}
&\leq\sqrt{1+\frac{2n(p+q)\alpha\epsilon_1}{\frac{1}{2}n(p+q)m_0^2}+\frac{n(p+q)\epsilon_1^2}{\frac{1}{2}n(p+q)m_0^2}}-1
\\&+\sqrt{\frac{n(p+q)\epsilon_1^2}{\frac{1}{2}n(p+q)m_0^2}}
\\&\leq\sqrt{1+\frac{8\max_{ij}|M_{ij}|\epsilon_1}{m_0^2}+\frac{2\epsilon_1^2}{m_0^2}}-1+\frac{\sqrt{2}\epsilon_1}{m_0}.
\end{aligned}
\end{eqnarray*}
\end{small}

If $\max_{ij}|M_{ij}|$, $\max_{ij}|M_{ij}|/m_0=O(1)$, then $\epsilon_E=O(\epsilon_1)$. Similarly, we have $\epsilon_J=\epsilon_K= O(\epsilon_1+\epsilon_2)$ where $\epsilon_J$ and $\epsilon_k$ are the errors of $|\psi_J\rangle$ and $|\psi_K\rangle$ respectively.

Let $U_E$, $U_J$, $U_K$ denote the unitary operations of preparing the states $|\psi_E\rangle$, $|\psi_J\rangle$ and $|\psi_K\rangle$ respectively, and $T_E$, $T_J$, $T_K$ represent the complexity corresponding to them. Once $T_E, T_J, T_K$ are obtained, the complexity of \emph{\textbf{Step 2}} can be calculated easily by the used lemmas and theorem. As a conclusion, the complexity of designing the block-encoding of $\tilde{H}$ is $O(\kappa(a_E+s+T_E)\log^2(\frac{\kappa^{3/2}}{\epsilon_3})+T_J+T_K)$ where $\kappa$ is the condition number of $\rho_E$, $s=\log (p+q)$, $a_E=\log n+3+\log m_1$, $\log m_1$ is the number of qubits in the third and fourth registers when preparing $|\psi_E\rangle$, $\epsilon_3$ is the error of $\tilde{U}_{\rho_E}$ and $\tilde{U}_{\rho_E}$ is a block-encoding of $\rho_E^{-\frac{1}{2}}$. See TABLE \ref{tab:stc} for more details of the complexity of other steps.

In \emph{\textbf{Step 3}}, the complexity of preparing the state $\rho_0$ is $O(\log (p+q))$ and it can be ignored. By Theorem 3 in Ref.~\cite{ICALP2019}, we can implement $e^{i\tilde{H}t}$ with complexity $O((|8\kappa t|+\log\frac{1}{|2t|\cdot \epsilon_{\tilde{H}}})(\tilde{T}_E+T_J+T_K))$, where $\epsilon_{\tilde{H}}=32\kappa^{\frac{3}{2}}(\epsilon_3+2\kappa^{\frac{1}{2}}\epsilon_J))$, $\tilde{T}_E$ is the complexity of $\tilde{U}_{\rho_E}$.
Using $e^{i\tilde{H}t}$, the eigenvalues and eigenvectors of $H$ to accuracy $\epsilon_4$ can be obtained by applying quantum phase estimation to $\rho_0$ for time $t=O(\frac{tr(J)}{tr(E)}\cdot \frac{1}{\epsilon_4})$. The value of $\frac{tr(J)}{tr(E)}$ can be determined if we replace fixed-point quantum search by measurements in stages (1.7) and (2.7), and
\begin{eqnarray*}
\begin{aligned}
         \frac{tr(J)}{tr(E)}=\frac{\|\acute{J}\|_F^2}{\|\acute{E}\|_F^2}
        &\leq\frac{c(p+q)(n'\alpha)^2}{(1/2)n(p+q)(m_0)^2}
        \\&=O(\frac{n(\max_{ij}|M_{ij}|)^2}{m_0^2}).
\end{aligned}
\end{eqnarray*}
Next, we use the quantum search algorithm with query complexity $O(\sqrt{p+q})$ and $O(d)$ times of repetition is enough to get the first $d$ largest eigenvalues of $H$. Therefore, the complexity of \emph{\textbf{Step 3}} is $O(d\sqrt{p+q}(\frac{n(\max_{ij}|M_{ij}|)^2\kappa}{m_0^2\epsilon_4}+\log\frac{m_0^2\epsilon_4}{n(\max_{ij}|M_{ij}|)^2\epsilon_{\tilde{H}}})(\tilde{T}_E+T_J+T_K))$.

According to Theorem 10 in Ref.~\cite{ICALP2019}, the complexity of \emph{\textbf{Step 4}} is $O(\kappa T_{1-3}\log \kappa)$, where $T_{1-3}$ is the total complexity from \emph{\textbf{Step 1}} to \emph{\textbf{Step 3}}.

If $d$, $\max_{ij}|M_{ij}|$, $\max_{ij}|M_{ij}|/m_0=O(1)$ and let $1/\epsilon_1,1/\epsilon_2,1/\epsilon_4,\kappa=O(\log (n(p+q))$, the complexity of the QDCCA algorithm can be reduced to $\widetilde{O}(n\sqrt{p+q})$. Note that with $\widetilde{O}$ we hide polylogarithmic factors. Since the complexity of the classical DCCA algorithm is $O(n(p+q)+(p+q)^3)$, our quantum algorithm achieves a polynomial speedup in the dimension of samples over the classical algorithm.

%where we use the notation $\widetilde{O}(f(x))$ to indicate $O(f(x)\mathrm {polylog}(f(x))$

\section{conclusion}\label{Sec:con}

In conclusion, we have proposed a QDCCA algorithm with rigorous complexity analysis. It has been shown that our quantum algorithm achieves a polynomial acceleration on the dimension of samples  over its classical counterpart when $d$, $\max_{ij}|M_{ij}|$, $\max_{ij}|M_{ij}|/m_0=O(1)$ and $1/\epsilon_1,1/\epsilon_2,1/\epsilon_4,\kappa=O(\log (n(p+q))$. The \emph{Lemma}~\ref{lem:ME} presented an efficient method to compute the mean of elements in a row of any real matrix, which can be reused as a subroutine for other quantum algorithms. Moreover, in the QDCCA algorithm, we completed the mean-centering when we prepared the density operators. It can be a separate quantum algorithm (called QMS algorithm) if we combine \emph{Lemma}~\ref{lem:ME} with QMA to realize $|i\rangle|j\rangle|0\rangle\rightarrow |i\rangle|j\rangle|\tilde{M}_{ij}\rangle$ where $\tilde{M}_{ij}=M_{ij}-\bar{M}_{i,*}$. The QMS algorithm achieves an exponential speedup in both the number of samples and their dimension over the classical MS algorithm. We can also modify the QMS algorithm to perform other data preprocessing operations, for example, Z-score standardization~\cite{HAN201283}. We hope that the techniques we presented in this paper will inspire others to explore more potential quantum algorithm applications in the future, such as expediting other classical preprocessing operations, solving the generalized eigenvalue problem under certain circumstances.

\section*{Acknowledgements}
We thank Linchun Wan and Mingchao Guo for useful discussions on the subject. This work is supported by National Natural Science Foundation of China (Grant Nos.  61976024, 61972048) and Beijing Natural Science Foundation (Grant No. 4222031).

\appendix

\section{Proof of \emph{Lemma} \ref{lem:ME}}
\label{app:proof}
Let us start by describing a procedure $U_y$ to estimate $\bar{L}_{i,*}$ of matrix $L$, and the idea behind $U_y$ is to calculate the mean by the inner product. We start with the initial state $|i\rangle_1|0\rangle_2|0\rangle_3^{\otimes\log d_2}|0\rangle_4$, the processes of $U_y$ are as follows.

(1) Perform a Hadamard gate on the second register, then apply $H^{\otimes \log d_2}$ to the third register to get
\begin{eqnarray*}
\begin{aligned}
|i\rangle_1\frac{1}{\sqrt{2}}(|0\rangle_2+|1\rangle_2)\frac{1}{\sqrt{d_2}}\sum_{j=1}^{d_2}|j\rangle_3|0\rangle_4.
\end{aligned}
\end{eqnarray*}

(2) Considering the second register as the control register, we perform controlled $\textbf{O}_L$ on the first, third and fourth registers to get
\begin{eqnarray*}
\begin{aligned}
&|i\rangle_1\frac{1}{\sqrt{2}}\bigg(|0\rangle_2\frac{1}{\sqrt{d_2}}\sum_{j=1}^{d_2}|j\rangle_3|L_{ij}\rangle_4
\\&+|1\rangle_2\frac{1}{\sqrt{d_2}}\sum_{j=1}^{d_2}|j\rangle_3|0\rangle_4\bigg).
\end{aligned}
\end{eqnarray*}

(3) Append an ancillary qubit $|0\rangle$ and then perform a appropriate controlled rotation on the ancillary qubit to get
\begin{eqnarray*}
\begin{aligned}
&|i\rangle_1\frac{1}{\sqrt{2}}\bigg[|0\rangle_2\frac{1}{\sqrt{d_2}}\sum_{j=1}^{d_2}|j\rangle_3|L_{ij}\rangle_4\bigg(\frac{L_{ij}}{C}|0\rangle_5
\\&+\sqrt{1-(\frac{L_{ij}}{C})^2}|1\rangle_5\bigg)
+|1\rangle_2\frac{1}{\sqrt{d_2}}\sum_{j=1}^{d_2}|j\rangle_3|0\rangle_4|0\rangle_5\bigg],
\end{aligned}
\end{eqnarray*}
where $C=\max_{ij}|L_{ij}|$.

(4) Uncompute the fourth register and let $|\phi_i\rangle :=\frac{1}{\sqrt{d_2}}\sum_{j=1}^{d_2}|j\rangle_3|0\rangle_4\bigg(\frac{L_{ij}}{C}|0\rangle_5
+\sqrt{1-(\frac{L_{ij}}{C})^2}|1\rangle_5\bigg)$, $|\varphi\rangle :=\frac{1}{\sqrt{d_2}}\sum_{j=1}^{d_2}|j\rangle_3|0\rangle_4|0\rangle_5$, then perform a Hadamard gate on the second register to get
\begin{eqnarray*}
\begin{aligned}
|i\rangle_1\bigg[\frac{1}{2}|0\rangle_2\bigg(|\phi_i\rangle+|\varphi\rangle\bigg)
+\frac{1}{2}|1\rangle_2\bigg(|\phi_i\rangle-|\varphi\rangle\bigg)\bigg].
\end{aligned}
\end{eqnarray*}

The probability of obtaining $|1\rangle$ when the second register is measured is  $P_{i1}=\frac{1-\langle\phi_i|\varphi\rangle}{2}$. It is obviously that $\langle\phi_i|\varphi\rangle=\frac{\sum_{j=1}^{d_2}L_{ij}}{d_2C}=\frac{\bar{L}_{i,*}}{C}, i=1,2,...,d_1$.

By swapping the registers, we can rewrite $|1\rangle_2(|\phi_i\rangle-|\varphi\rangle)$ as $|y_i,1\rangle$, and hence we have the final mapping
\begin{eqnarray*}
\begin{aligned}
U_{y}: |i\rangle|0\rangle\rightarrow|i\rangle\bigg(\sqrt{P_{i1}}|y_i,1\rangle+\sqrt{1-P_{i1}}|G_i,0\rangle\bigg)\ \ \
\end{aligned}
\end{eqnarray*}
which can be carried out in time $O(\log (d_1d_2))$, where $|G_i\rangle$ is a garbage state.

Then, similar to Ref.~\cite{qmeans2019}, we can use $U_{y}$, amplitude estimation~\cite{BHMT} and median evaluation~\cite{WKS2015} to get a quantum state $|\psi_i\rangle$ for any $\Delta>0$ such that,
\begin{eqnarray*}
\begin{aligned}
\parallel|\psi_i\rangle-|0\rangle^{\otimes ls}|\tilde{P_{i1}},G\rangle\parallel_2\leq\sqrt{2\Delta},
\end{aligned}
\end{eqnarray*}
where $l$ is an integer, $s$ is the number of qubits in $|\tilde{P_{i1}},y_i,1\rangle$, $|\tilde{P_{i1}}-P_{i1}|\leq \epsilon$ and $|G\rangle$ is a garbage register.
The running time of the procedure is $O(\frac{\log (d_1d_2)\log\frac{1}{\Delta}}{\epsilon})$.

Finally, we can easily compute $\bar{L}_{i,*}=C(1-2\tilde{P_{i1}})$. If we want to have in the end an absolute error $\epsilon$, we should control the error of amplitude estimation as $\frac{\epsilon}{2C}$. Therefore, the total time complexity of $U_{mean}$ is $O(\frac{\max_{ij}|L_{ij}| \log (d_1d_2)\log\frac{1}{\Delta}}{\epsilon})$ where $\epsilon$ is the error of $\bar{L}_{i,*}$.

This concludes the proof of \emph{Lemma}~\ref{lem:ME}.
\\
\section{Parameters analysis }
\label{app:pa}

In this Appendix we analyze the choice of parameters $\alpha$ and $\beta$ in parts (1) and (2) respectively.

(1) The choice of parameter $\alpha$. If we want to perform the controlled rotation effectively, the condition of $\alpha\geq \max_{ij}|\acute{E}_{ij}|$ must be satisfied. Moreover,
\begin{eqnarray*}
    \begin{aligned}
        \max_{ij}|\acute{E}_{ij}|
        &=\max_{ij}|M_{ij}-\bar{M}_{i,*}|
        \\&\leq\max_{ij}(|M_{ij}|+|\bar{M}_{i,*}|)
        \\&\leq2\max_{ij}|M_{ij}|.
    \end{aligned}
    \end{eqnarray*}
Then, we can chose $\alpha=2\max_{ij}|M_{ij}|$ to make sure that $\frac{\acute{E}_{ij}}{\alpha}$ is no more than 1.

(2) The choice of parameter $\beta$. Note that $M$ can be rewritten as $[M^1,...,M^c]$ where $M^i=\left(\begin{matrix}
   A^i \\
   B^i\\
  \end{matrix}
  \right), i=1,...,c,$
  and
\begin{eqnarray*}
    \begin{aligned}
        \max_{ki}|\acute{J}_{ki}|
        &=\max_{ikj}\bigg(\sum_{j}|M_{kj}^i-\bar{M}_{k,*}|\bigg)
        \\&\leq\max_{ikj}\bigg[\sum_{j}\bigg(|M_{kj}^i|+|\bar{M}_{k,*}|\bigg)\bigg]
        \\&\leq2n'\max_{ij}|M_{ij}|.
    \end{aligned}
    \end{eqnarray*}
Hence, we can chose $\beta=2n'\max_{ij}|M_{ij}|$.

\bibliography{refe}

\begin{thebibliography}{40}%
\makeatletter
\providecommand \@ifxundefined [1]{%
 \@ifx{#1\undefined}
}%
\providecommand \@ifnum [1]{%
 \ifnum #1\expandafter \@firstoftwo
 \else \expandafter \@secondoftwo
 \fi
}%
\providecommand \@ifx [1]{%
 \ifx #1\expandafter \@firstoftwo
 \else \expandafter \@secondoftwo
 \fi
}%
\providecommand \natexlab [1]{#1}%
\providecommand \enquote  [1]{``#1''}%
\providecommand \bibnamefont  [1]{#1}%
\providecommand \bibfnamefont [1]{#1}%
\providecommand \citenamefont [1]{#1}%
\providecommand \href@noop [0]{\@secondoftwo}%
\providecommand \href [0]{\begingroup \@sanitize@url \@href}%
\providecommand \@href[1]{\@@startlink{#1}\@@href}%
\providecommand \@@href[1]{\endgroup#1\@@endlink}%
\providecommand \@sanitize@url [0]{\catcode `\\12\catcode `\$12\catcode
  `\&12\catcode `\#12\catcode `\^12\catcode `\_12\catcode `\%12\relax}%
\providecommand \@@startlink[1]{}%
\providecommand \@@endlink[0]{}%
\providecommand \url  [0]{\begingroup\@sanitize@url \@url }%
\providecommand \@url [1]{\endgroup\@href {#1}{\urlprefix }}%
\providecommand \urlprefix  [0]{URL }%
\providecommand \Eprint [0]{\href }%
\providecommand \doibase [0]{http://dx.doi.org/}%
\providecommand \selectlanguage [0]{\@gobble}%
\providecommand \bibinfo  [0]{\@secondoftwo}%
\providecommand \bibfield  [0]{\@secondoftwo}%
\providecommand \translation [1]{[#1]}%
\providecommand \BibitemOpen [0]{}%
\providecommand \bibitemStop [0]{}%
\providecommand \bibitemNoStop [0]{.\EOS\space}%
\providecommand \EOS [0]{\spacefactor3000\relax}%
\providecommand \BibitemShut  [1]{\csname bibitem#1\endcsname}%
\let\auto@bib@innerbib\@empty
%</preamble>
\bibitem [{\citenamefont {{Shor}}(1994)}]{PS1994}%
  \BibitemOpen
  \bibfield  {author} {\bibinfo {author} {\bibfnamefont {P.~W.}\ \bibnamefont
  {{Shor}}},\ }in\ \href@noop {} {\emph {\bibinfo {booktitle} {Proceedings 35th
  Annual Symposium on Foundations of Computer Science}}}\ (\bibinfo {year}
  {1994})\ pp.\ \bibinfo {pages} {124--134}\BibitemShut {NoStop}%
\bibitem [{\citenamefont {Grover}(1996)}]{GL1996}%
  \BibitemOpen
  \bibfield  {author} {\bibinfo {author} {\bibfnamefont {L.~K.}\ \bibnamefont
  {Grover}},\ }in\ \href@noop {} {\emph {\bibinfo {booktitle} {Proceedings of
  the Twenty-Eighth Annual ACM Symposium on Theory of Computing}}},\ \bibinfo
  {series and number} {STOC ¡¯96}\ (\bibinfo {year} {1996})\ pp.\ \bibinfo
  {pages} {212--219}\BibitemShut {NoStop}%
\bibitem [{\citenamefont {Rebentrost}\ \emph {et~al.}(2014)\citenamefont
  {Rebentrost}, \citenamefont {Mohseni},\ and\ \citenamefont {Lloyd}}]{RML}%
  \BibitemOpen
  \bibfield  {author} {\bibinfo {author} {\bibfnamefont {P.}~\bibnamefont
  {Rebentrost}}, \bibinfo {author} {\bibfnamefont {M.}~\bibnamefont {Mohseni}},
  \ and\ \bibinfo {author} {\bibfnamefont {S.}~\bibnamefont {Lloyd}},\
  }\href@noop {} {\bibfield  {journal} {\bibinfo  {journal} {Phys. Rev. Lett.}\
  }\textbf {\bibinfo {volume} {113}},\ \bibinfo {pages} {130503} (\bibinfo
  {year} {2014})}\BibitemShut {NoStop}%
\bibitem [{\citenamefont {P{\'{e}}rez-Salinas}\ \emph
  {et~al.}(2020)\citenamefont {P{\'{e}}rez-Salinas}, \citenamefont
  {Cervera-Lierta}, \citenamefont {Gil-Fuster},\ and\ \citenamefont
  {Latorre}}]{QC1}%
  \BibitemOpen
  \bibfield  {author} {\bibinfo {author} {\bibfnamefont {A.}~\bibnamefont
  {P{\'{e}}rez-Salinas}}, \bibinfo {author} {\bibfnamefont {A.}~\bibnamefont
  {Cervera-Lierta}}, \bibinfo {author} {\bibfnamefont {E.}~\bibnamefont
  {Gil-Fuster}}, \ and\ \bibinfo {author} {\bibfnamefont {J.~I.}\ \bibnamefont
  {Latorre}},\ }\href {\doibase 10.22331/q-2020-02-06-226} {\bibfield
  {journal} {\bibinfo  {journal} {{Quantum}}\ }\textbf {\bibinfo {volume}
  {4}},\ \bibinfo {pages} {226} (\bibinfo {year} {2020})}\BibitemShut {NoStop}%
\bibitem [{\citenamefont {Du}\ \emph {et~al.}(2021)\citenamefont {Du},
  \citenamefont {Hsieh}, \citenamefont {Liu},\ and\ \citenamefont {Tao}}]{QC2}%
  \BibitemOpen
  \bibfield  {author} {\bibinfo {author} {\bibfnamefont {Y.}~\bibnamefont
  {Du}}, \bibinfo {author} {\bibfnamefont {M.-H.}\ \bibnamefont {Hsieh}},
  \bibinfo {author} {\bibfnamefont {T.}~\bibnamefont {Liu}}, \ and\ \bibinfo
  {author} {\bibfnamefont {D.}~\bibnamefont {Tao}},\ }\href {\doibase
  10.1088/1367-2630/abdefa} {\bibfield  {journal} {\bibinfo  {journal} {New
  Journal of Physics}\ }\textbf {\bibinfo {volume} {23}},\ \bibinfo {pages}
  {023020} (\bibinfo {year} {2021})}\BibitemShut {NoStop}%
\bibitem [{\citenamefont {Huang}\ \emph {et~al.}(2021)\citenamefont {Huang},
  \citenamefont {Tan},\ and\ \citenamefont {Xu}}]{jinan}%
  \BibitemOpen
  \bibfield  {author} {\bibinfo {author} {\bibfnamefont {R.}~\bibnamefont
  {Huang}}, \bibinfo {author} {\bibfnamefont {X.-Q.}\ \bibnamefont {Tan}}, \
  and\ \bibinfo {author} {\bibfnamefont {Q.-S.}\ \bibnamefont {Xu}},\ }\href
  {\doibase https://doi.org/10.1016/j.neucom.2021.04.074} {\bibfield  {journal}
  {\bibinfo  {journal} {Neurocomputing}\ }\textbf {\bibinfo {volume} {452}},\
  \bibinfo {pages} {89} (\bibinfo {year} {2021})}\BibitemShut {NoStop}%
\bibitem [{\citenamefont {Wang}(2017)}]{LR1}%
  \BibitemOpen
  \bibfield  {author} {\bibinfo {author} {\bibfnamefont {G.}~\bibnamefont
  {Wang}},\ }\href {\doibase 10.1103/PhysRevA.96.012335} {\bibfield  {journal}
  {\bibinfo  {journal} {Phys. Rev. A}\ }\textbf {\bibinfo {volume} {96}},\
  \bibinfo {pages} {012335} (\bibinfo {year} {2017})}\BibitemShut {NoStop}%
\bibitem [{\citenamefont {Yu}\ \emph {et~al.}(2021)\citenamefont {Yu},
  \citenamefont {Gao},\ and\ \citenamefont {Wen}}]{LR2}%
  \BibitemOpen
  \bibfield  {author} {\bibinfo {author} {\bibfnamefont {C.-H.}\ \bibnamefont
  {Yu}}, \bibinfo {author} {\bibfnamefont {F.}~\bibnamefont {Gao}}, \ and\
  \bibinfo {author} {\bibfnamefont {Q.-Y.}\ \bibnamefont {Wen}},\ }\href@noop
  {} {\bibfield  {journal} {\bibinfo  {journal} {IEEE Transactions on Knowledge
  and Data Engineering}\ }\textbf {\bibinfo {volume} {33}},\ \bibinfo {pages}
  {858} (\bibinfo {year} {2021})}\BibitemShut {NoStop}%
\bibitem [{\citenamefont {Yu}\ \emph {et~al.}(2019)\citenamefont {Yu},
  \citenamefont {Gao}, \citenamefont {Liu}, \citenamefont {Huynh},
  \citenamefont {Reynolds},\ and\ \citenamefont {Wang}}]{LR3}%
  \BibitemOpen
  \bibfield  {author} {\bibinfo {author} {\bibfnamefont {C.-H.}\ \bibnamefont
  {Yu}}, \bibinfo {author} {\bibfnamefont {F.}~\bibnamefont {Gao}}, \bibinfo
  {author} {\bibfnamefont {C.}~\bibnamefont {Liu}}, \bibinfo {author}
  {\bibfnamefont {D.}~\bibnamefont {Huynh}}, \bibinfo {author} {\bibfnamefont
  {M.}~\bibnamefont {Reynolds}}, \ and\ \bibinfo {author} {\bibfnamefont
  {J.}~\bibnamefont {Wang}},\ }\href {\doibase 10.1103/PhysRevA.99.022301}
  {\bibfield  {journal} {\bibinfo  {journal} {Phys. Rev. A}\ }\textbf {\bibinfo
  {volume} {99}},\ \bibinfo {pages} {022301} (\bibinfo {year}
  {2019})}\BibitemShut {NoStop}%
\bibitem [{\citenamefont {Lloyd}\ \emph {et~al.}(2014)\citenamefont {Lloyd},
  \citenamefont {Mohseni},\ and\ \citenamefont {Rebentrost}}]{QPCA}%
  \BibitemOpen
  \bibfield  {author} {\bibinfo {author} {\bibfnamefont {S.}~\bibnamefont
  {Lloyd}}, \bibinfo {author} {\bibfnamefont {M.}~\bibnamefont {Mohseni}}, \
  and\ \bibinfo {author} {\bibfnamefont {P.}~\bibnamefont {Rebentrost}},\
  }\href {\doibase 10.1038/nphys3029} {\bibfield  {journal} {\bibinfo
  {journal} {Nature Physics}\ }\textbf {\bibinfo {volume} {10}},\ \bibinfo
  {pages} {631} (\bibinfo {year} {2014})}\BibitemShut {NoStop}%
\bibitem [{\citenamefont {Cong}\ and\ \citenamefont {Duan}(2016)}]{QLDA}%
  \BibitemOpen
  \bibfield  {author} {\bibinfo {author} {\bibfnamefont {I.}~\bibnamefont
  {Cong}}\ and\ \bibinfo {author} {\bibfnamefont {L.}~\bibnamefont {Duan}},\
  }\href@noop {} {\bibfield  {journal} {\bibinfo  {journal} {New Journal of
  Physics}\ }\textbf {\bibinfo {volume} {18}},\ \bibinfo {pages} {073011}
  (\bibinfo {year} {2016})}\BibitemShut {NoStop}%
\bibitem [{\citenamefont {Duan}\ \emph {et~al.}(2019)\citenamefont {Duan},
  \citenamefont {Yuan}, \citenamefont {Xu},\ and\ \citenamefont {Li}}]{DYXL}%
  \BibitemOpen
  \bibfield  {author} {\bibinfo {author} {\bibfnamefont {B.}~\bibnamefont
  {Duan}}, \bibinfo {author} {\bibfnamefont {J.}~\bibnamefont {Yuan}}, \bibinfo
  {author} {\bibfnamefont {J.}~\bibnamefont {Xu}}, \ and\ \bibinfo {author}
  {\bibfnamefont {D.}~\bibnamefont {Li}},\ }\href@noop {} {\bibfield  {journal}
  {\bibinfo  {journal} {Phys. Rev. A}\ }\textbf {\bibinfo {volume} {99}},\
  \bibinfo {pages} {032311} (\bibinfo {year} {2019})}\BibitemShut {NoStop}%
\bibitem [{\citenamefont {Pan}\ \emph {et~al.}(2022)\citenamefont {Pan},
  \citenamefont {Wan}, \citenamefont {Liu}, \citenamefont {Wu}, \citenamefont
  {Qin}, \citenamefont {Wen},\ and\ \citenamefont {Gao}}]{pan2022}%
  \BibitemOpen
  \bibfield  {author} {\bibinfo {author} {\bibfnamefont {S.-J.}\ \bibnamefont
  {Pan}}, \bibinfo {author} {\bibfnamefont {L.-C.}\ \bibnamefont {Wan}},
  \bibinfo {author} {\bibfnamefont {H.-L.}\ \bibnamefont {Liu}}, \bibinfo
  {author} {\bibfnamefont {Y.-S.}\ \bibnamefont {Wu}}, \bibinfo {author}
  {\bibfnamefont {S.-J.}\ \bibnamefont {Qin}}, \bibinfo {author} {\bibfnamefont
  {Q.-Y.}\ \bibnamefont {Wen}}, \ and\ \bibinfo {author} {\bibfnamefont
  {F.}~\bibnamefont {Gao}},\ }\href
  {http://iopscience.iop.org/article/10.1088/1674-1056/ac523a} {\bibfield
  {journal} {\bibinfo  {journal} {Chinese Physics B}\ } (\bibinfo {year}
  {2022})}\BibitemShut {NoStop}%
\bibitem [{\citenamefont {Harrow}\ \emph {et~al.}(2009)\citenamefont {Harrow},
  \citenamefont {Hassidim},\ and\ \citenamefont {Lloyd}}]{HHL}%
  \BibitemOpen
  \bibfield  {author} {\bibinfo {author} {\bibfnamefont {A.~W.}\ \bibnamefont
  {Harrow}}, \bibinfo {author} {\bibfnamefont {A.}~\bibnamefont {Hassidim}}, \
  and\ \bibinfo {author} {\bibfnamefont {S.}~\bibnamefont {Lloyd}},\
  }\href@noop {} {\bibfield  {journal} {\bibinfo  {journal} {Phys. Rev. Lett.}\
  }\textbf {\bibinfo {volume} {103}},\ \bibinfo {pages} {150502} (\bibinfo
  {year} {2009})}\BibitemShut {NoStop}%
\bibitem [{\citenamefont {Wossnig}\ \emph {et~al.}(2018)\citenamefont
  {Wossnig}, \citenamefont {Zhao},\ and\ \citenamefont {Prakash}}]{DM2018}%
  \BibitemOpen
  \bibfield  {author} {\bibinfo {author} {\bibfnamefont {L.}~\bibnamefont
  {Wossnig}}, \bibinfo {author} {\bibfnamefont {Z.}~\bibnamefont {Zhao}}, \
  and\ \bibinfo {author} {\bibfnamefont {A.}~\bibnamefont {Prakash}},\ }\href
  {\doibase 10.1103/PhysRevLett.120.050502} {\bibfield  {journal} {\bibinfo
  {journal} {Phys. Rev. Lett.}\ }\textbf {\bibinfo {volume} {120}},\ \bibinfo
  {pages} {050502} (\bibinfo {year} {2018})}\BibitemShut {NoStop}%
\bibitem [{\citenamefont {Wan}\ \emph {et~al.}(2018)\citenamefont {Wan},
  \citenamefont {Yu}, \citenamefont {Pan}, \citenamefont {Gao}, \citenamefont
  {Wen},\ and\ \citenamefont {Qin}}]{wan2018}%
  \BibitemOpen
  \bibfield  {author} {\bibinfo {author} {\bibfnamefont {L.-C.}\ \bibnamefont
  {Wan}}, \bibinfo {author} {\bibfnamefont {C.-H.}\ \bibnamefont {Yu}},
  \bibinfo {author} {\bibfnamefont {S.-J.}\ \bibnamefont {Pan}}, \bibinfo
  {author} {\bibfnamefont {F.}~\bibnamefont {Gao}}, \bibinfo {author}
  {\bibfnamefont {Q.-Y.}\ \bibnamefont {Wen}}, \ and\ \bibinfo {author}
  {\bibfnamefont {S.-J.}\ \bibnamefont {Qin}},\ }\href {\doibase
  10.1103/PhysRevA.97.062322} {\bibfield  {journal} {\bibinfo  {journal} {Phys.
  Rev. A}\ }\textbf {\bibinfo {volume} {97}},\ \bibinfo {pages} {062322}
  (\bibinfo {year} {2018})}\BibitemShut {NoStop}%
\bibitem [{\citenamefont {Liu}\ \emph {et~al.}(2022)\citenamefont {Liu},
  \citenamefont {Qin}, \citenamefont {Wan}, \citenamefont {Yu}, \citenamefont
  {Pan}, \citenamefont {Gao},\ and\ \citenamefont {Wen}}]{liu}%
  \BibitemOpen
  \bibfield  {author} {\bibinfo {author} {\bibfnamefont {H.-L.}\ \bibnamefont
  {Liu}}, \bibinfo {author} {\bibfnamefont {S.-J.}\ \bibnamefont {Qin}},
  \bibinfo {author} {\bibfnamefont {L.-C.}\ \bibnamefont {Wan}}, \bibinfo
  {author} {\bibfnamefont {C.-H.}\ \bibnamefont {Yu}}, \bibinfo {author}
  {\bibfnamefont {S.-J.}\ \bibnamefont {Pan}}, \bibinfo {author} {\bibfnamefont
  {F.}~\bibnamefont {Gao}}, \ and\ \bibinfo {author} {\bibfnamefont {Q.-Y.}\
  \bibnamefont {Wen}},\ }\href@noop {} {\bibfield  {journal} {\bibinfo
  {journal} {arXiv:2203.14451v1}\ } (\bibinfo {year} {2022})}\BibitemShut
  {NoStop}%
\bibitem [{\citenamefont {Sun}\ \emph {et~al.}(2008)\citenamefont {Sun},
  \citenamefont {Chen}, \citenamefont {Yang},\ and\ \citenamefont
  {Shi}}]{DCCA}%
  \BibitemOpen
  \bibfield  {author} {\bibinfo {author} {\bibfnamefont {T.-K.}\ \bibnamefont
  {Sun}}, \bibinfo {author} {\bibfnamefont {S.-C.}\ \bibnamefont {Chen}},
  \bibinfo {author} {\bibfnamefont {J.-Y.}\ \bibnamefont {Yang}}, \ and\
  \bibinfo {author} {\bibfnamefont {P.-F.}\ \bibnamefont {Shi}},\ }in\
  \href@noop {} {\emph {\bibinfo {booktitle} {2008 Eighth IEEE International
  Conference on Data Mining}}}\ (\bibinfo {year} {2008})\ pp.\ \bibinfo {pages}
  {1043--1048}\BibitemShut {NoStop}%
\bibitem [{\citenamefont {Yang}\ \emph {et~al.}(2021)\citenamefont {Yang},
  \citenamefont {Liu}, \citenamefont {Liu},\ and\ \citenamefont {Tao}}]{SCCA}%
  \BibitemOpen
  \bibfield  {author} {\bibinfo {author} {\bibfnamefont {X.-H.}\ \bibnamefont
  {Yang}}, \bibinfo {author} {\bibfnamefont {W.-F.}\ \bibnamefont {Liu}},
  \bibinfo {author} {\bibfnamefont {W.}~\bibnamefont {Liu}}, \ and\ \bibinfo
  {author} {\bibfnamefont {D.-C.}\ \bibnamefont {Tao}},\ }\href@noop {}
  {\bibfield  {journal} {\bibinfo  {journal} {IEEE Transactions on Knowledge
  and Data Engineering}\ }\textbf {\bibinfo {volume} {33}},\ \bibinfo {pages}
  {2349} (\bibinfo {year} {2021})}\BibitemShut {NoStop}%
\bibitem [{\citenamefont {Sargin}\ \emph {et~al.}(2006)\citenamefont {Sargin},
  \citenamefont {Erzin}, \citenamefont {Yemez},\ and\ \citenamefont
  {Tekalp}}]{CCA1}%
  \BibitemOpen
  \bibfield  {author} {\bibinfo {author} {\bibfnamefont {M.~E.}\ \bibnamefont
  {Sargin}}, \bibinfo {author} {\bibfnamefont {E.}~\bibnamefont {Erzin}},
  \bibinfo {author} {\bibfnamefont {Y.}~\bibnamefont {Yemez}}, \ and\ \bibinfo
  {author} {\bibfnamefont {A.~M.}\ \bibnamefont {Tekalp}},\ }\href@noop {}
  {\bibfield  {journal} {\bibinfo  {journal} {2006 IEEE International
  Conference on Acoustics Speech and Signal Processing Proceedings}\ }\textbf
  {\bibinfo {volume} {1}},\ \bibinfo {pages} {I} (\bibinfo {year}
  {2006})}\BibitemShut {NoStop}%
\bibitem [{\citenamefont {Sun}\ \emph {et~al.}(2005)\citenamefont {Sun},
  \citenamefont {Zeng}, \citenamefont {Liu}, \citenamefont {Heng},\ and\
  \citenamefont {Xia}}]{CCA2}%
  \BibitemOpen
  \bibfield  {author} {\bibinfo {author} {\bibfnamefont {Q.-S.}\ \bibnamefont
  {Sun}}, \bibinfo {author} {\bibfnamefont {S.-G.}\ \bibnamefont {Zeng}},
  \bibinfo {author} {\bibfnamefont {Y.}~\bibnamefont {Liu}}, \bibinfo {author}
  {\bibfnamefont {P.-A.}\ \bibnamefont {Heng}}, \ and\ \bibinfo {author}
  {\bibfnamefont {D.-S.}\ \bibnamefont {Xia}},\ }\href {\doibase
  https://doi.org/10.1016/j.patcog.2004.12.013} {\bibfield  {journal} {\bibinfo
   {journal} {Pattern Recognition}\ }\textbf {\bibinfo {volume} {38}},\
  \bibinfo {pages} {2437} (\bibinfo {year} {2005})}\BibitemShut {NoStop}%
\bibitem [{\citenamefont {Wegelin}(2000)}]{PLS}%
  \BibitemOpen
  \bibfield  {author} {\bibinfo {author} {\bibfnamefont {J.}~\bibnamefont
  {Wegelin}},\ }\href@noop {} {\bibfield  {journal} {\bibinfo  {journal}
  {Technical report}\ } (\bibinfo {year} {2000})}\BibitemShut {NoStop}%
\bibitem [{\citenamefont {Koide-Majima}\ and\ \citenamefont
  {Majima}(2021)}]{qiCCA}%
  \BibitemOpen
  \bibfield  {author} {\bibinfo {author} {\bibfnamefont {N.}~\bibnamefont
  {Koide-Majima}}\ and\ \bibinfo {author} {\bibfnamefont {K.}~\bibnamefont
  {Majima}},\ }\href {\doibase https://doi.org/10.1016/j.neunet.2020.11.019}
  {\bibfield  {journal} {\bibinfo  {journal} {Neural Networks}\ }\textbf
  {\bibinfo {volume} {135}},\ \bibinfo {pages} {55} (\bibinfo {year}
  {2021})}\BibitemShut {NoStop}%
\bibitem [{\citenamefont {Hou}\ \emph {et~al.}(2021)\citenamefont {Hou},
  \citenamefont {Li}, \citenamefont {Chen},\ and\ \citenamefont {Tian}}]{QPLS}%
  \BibitemOpen
  \bibfield  {author} {\bibinfo {author} {\bibfnamefont {Y.-Y.}\ \bibnamefont
  {Hou}}, \bibinfo {author} {\bibfnamefont {J.}~\bibnamefont {Li}}, \bibinfo
  {author} {\bibfnamefont {X.-B.}\ \bibnamefont {Chen}}, \ and\ \bibinfo
  {author} {\bibfnamefont {Y.}~\bibnamefont {Tian}},\ }\href
  {http://iopscience.iop.org/article/10.1088/1674-1056/ac1b84} {\bibfield
  {journal} {\bibinfo  {journal} {Chinese Physics B}\ } (\bibinfo {year}
  {2021})}\BibitemShut {NoStop}%
\bibitem [{\citenamefont {Kerenidis}\ \emph {et~al.}(2019)\citenamefont
  {Kerenidis}, \citenamefont {Landman}, \citenamefont {Luongo},\ and\
  \citenamefont {Prakash}}]{qmeans2019}%
  \BibitemOpen
  \bibfield  {author} {\bibinfo {author} {\bibfnamefont {I.}~\bibnamefont
  {Kerenidis}}, \bibinfo {author} {\bibfnamefont {J.}~\bibnamefont {Landman}},
  \bibinfo {author} {\bibfnamefont {A.}~\bibnamefont {Luongo}}, \ and\ \bibinfo
  {author} {\bibfnamefont {A.}~\bibnamefont {Prakash}},\ }in\ \href
  {https://proceedings.neurips.cc/paper/2019/file/16026d60ff9b54410b3435b403afd226-Paper.pdf}
  {\emph {\bibinfo {booktitle} {Advances in Neural Information Processing
  Systems}}},\ Vol.~\bibinfo {volume} {32},\ \bibinfo {editor} {edited by\
  \bibinfo {editor} {\bibfnamefont {H.}~\bibnamefont {Wallach}}, \bibinfo
  {editor} {\bibfnamefont {H.}~\bibnamefont {Larochelle}}, \bibinfo {editor}
  {\bibfnamefont {A.}~\bibnamefont {Beygelzimer}}, \bibinfo {editor}
  {\bibfnamefont {F.}~\bibnamefont {Alch\'{e}-Buc}}, \bibinfo {editor}
  {\bibfnamefont {E.}~\bibnamefont {Fox}}, \ and\ \bibinfo {editor}
  {\bibfnamefont {R.}~\bibnamefont {Garnett}}}\ (\bibinfo  {publisher} {Curran
  Associates, Inc.},\ \bibinfo {year} {2019})\BibitemShut {NoStop}%
\bibitem [{\citenamefont {Brassard}\ \emph {et~al.}(2002)\citenamefont
  {Brassard}, \citenamefont {Hoyer}, \citenamefont {Mosca},\ and\ \citenamefont
  {Tapp}}]{BHMT}%
  \BibitemOpen
  \bibfield  {author} {\bibinfo {author} {\bibfnamefont {G.}~\bibnamefont
  {Brassard}}, \bibinfo {author} {\bibfnamefont {P.}~\bibnamefont {Hoyer}},
  \bibinfo {author} {\bibfnamefont {M.}~\bibnamefont {Mosca}}, \ and\ \bibinfo
  {author} {\bibfnamefont {A.}~\bibnamefont {Tapp}},\ }\href@noop {} {\bibfield
   {journal} {\bibinfo  {journal} {Contemporary Mathematics}\ }\textbf
  {\bibinfo {volume} {305}},\ \bibinfo {pages} {53} (\bibinfo {year}
  {2002})}\BibitemShut {NoStop}%
\bibitem [{\citenamefont {Wiebe}\ \emph {et~al.}(2015)\citenamefont {Wiebe},
  \citenamefont {Kapoor},\ and\ \citenamefont {Svore}}]{WKS2015}%
  \BibitemOpen
  \bibfield  {author} {\bibinfo {author} {\bibfnamefont {N.}~\bibnamefont
  {Wiebe}}, \bibinfo {author} {\bibfnamefont {A.}~\bibnamefont {Kapoor}}, \
  and\ \bibinfo {author} {\bibfnamefont {K.~M.}\ \bibnamefont {Svore}},\
  }\href@noop {} {\bibfield  {journal} {\bibinfo  {journal} {Quantum Info.
  Comput.}\ }\textbf {\bibinfo {volume} {15}},\ \bibinfo {pages} {316}
  (\bibinfo {year} {2015})}\BibitemShut {NoStop}%
\bibitem [{\citenamefont {Zhou}\ \emph {et~al.}(2017)\citenamefont {Zhou},
  \citenamefont {Loke}, \citenamefont {Izaac},\ and\ \citenamefont
  {Wang}}]{QMA}%
  \BibitemOpen
  \bibfield  {author} {\bibinfo {author} {\bibfnamefont {S.~S.}\ \bibnamefont
  {Zhou}}, \bibinfo {author} {\bibfnamefont {T.}~\bibnamefont {Loke}}, \bibinfo
  {author} {\bibfnamefont {J.~A.}\ \bibnamefont {Izaac}}, \ and\ \bibinfo
  {author} {\bibfnamefont {J.~B.}\ \bibnamefont {Wang}},\ }\href@noop {}
  {\bibfield  {journal} {\bibinfo  {journal} {Quantum Information Processing}\
  }\textbf {\bibinfo {volume} {16}},\ \bibinfo {pages} {82} (\bibinfo {year}
  {2017})}\BibitemShut {NoStop}%
\bibitem [{\citenamefont {Ruiz-Perez}\ and\ \citenamefont
  {Garcia-Escartin}(2017)}]{QMA2}%
  \BibitemOpen
  \bibfield  {author} {\bibinfo {author} {\bibfnamefont {L.}~\bibnamefont
  {Ruiz-Perez}}\ and\ \bibinfo {author} {\bibfnamefont {J.~C.}\ \bibnamefont
  {Garcia-Escartin}},\ }\href@noop {} {\bibfield  {journal} {\bibinfo
  {journal} {Quantum Information Processing}\ }\textbf {\bibinfo {volume}
  {16}},\ \bibinfo {pages} {152} (\bibinfo {year} {2017})}\BibitemShut
  {NoStop}%
\bibitem [{\citenamefont {Shao}\ and\ \citenamefont {Liu}(2020)}]{GEP3}%
  \BibitemOpen
  \bibfield  {author} {\bibinfo {author} {\bibfnamefont {C.-P.}\ \bibnamefont
  {Shao}}\ and\ \bibinfo {author} {\bibfnamefont {J.-P.}\ \bibnamefont {Liu}},\
  }\href@noop {} {\bibfield  {journal} {\bibinfo  {journal} {arXiv:2010.15027v1
  [quant-ph]}\ } (\bibinfo {year} {2020})}\BibitemShut {NoStop}%
\bibitem [{\citenamefont {Chakraborty}\ \emph {et~al.}(2019)\citenamefont
  {Chakraborty}, \citenamefont {Gily{\'e}n},\ and\ \citenamefont
  {Jeffery}}]{ICALP2019}%
  \BibitemOpen
  \bibfield  {author} {\bibinfo {author} {\bibfnamefont {S.}~\bibnamefont
  {Chakraborty}}, \bibinfo {author} {\bibfnamefont {A.}~\bibnamefont
  {Gily{\'e}n}}, \ and\ \bibinfo {author} {\bibfnamefont {S.}~\bibnamefont
  {Jeffery}},\ }in\ \href@noop {} {\emph {\bibinfo {booktitle} {46th
  International Colloquium on Automata, Languages, and Programming (ICALP
  2019)}}},\ \bibinfo {series} {Leibniz International Proceedings in
  Informatics (LIPIcs)}, Vol.\ \bibinfo {volume} {132}\ (\bibinfo {year}
  {2019})\ pp.\ \bibinfo {pages} {33:1--33:14}\BibitemShut {NoStop}%
\bibitem [{\citenamefont {Gily{\'e}n}\ \emph {et~al.}(2019)\citenamefont
  {Gily{\'e}n}, \citenamefont {Su}, \citenamefont {Low},\ and\ \citenamefont
  {Wiebe}}]{GSL}%
  \BibitemOpen
  \bibfield  {author} {\bibinfo {author} {\bibfnamefont {A.}~\bibnamefont
  {Gily{\'e}n}}, \bibinfo {author} {\bibfnamefont {Y.}~\bibnamefont {Su}},
  \bibinfo {author} {\bibfnamefont {G.~H.}\ \bibnamefont {Low}}, \ and\
  \bibinfo {author} {\bibfnamefont {N.}~\bibnamefont {Wiebe}},\ }in\ \href@noop
  {} {\emph {\bibinfo {booktitle} {Proceedings of the 51st Annual ACM SIGACT
  Symposium on Theory of Computing}}}\ (\bibinfo {year} {2019})\ pp.\ \bibinfo
  {pages} {193--204}\BibitemShut {NoStop}%
\bibitem [{\citenamefont {Chakraborty}\ \emph
  {et~al.}(2018{\natexlab{a}})\citenamefont {Chakraborty}, \citenamefont
  {Gily{\'e}n},\ and\ \citenamefont {Jeffery}}]{CGJ}%
  \BibitemOpen
  \bibfield  {author} {\bibinfo {author} {\bibfnamefont {S.}~\bibnamefont
  {Chakraborty}}, \bibinfo {author} {\bibfnamefont {A.}~\bibnamefont
  {Gily{\'e}n}}, \ and\ \bibinfo {author} {\bibfnamefont {S.}~\bibnamefont
  {Jeffery}},\ }\href@noop {} {\bibfield  {journal} {\bibinfo  {journal}
  {arXiv: 1804.01973v2}\ } (\bibinfo {year} {2018}{\natexlab{a}})}\BibitemShut
  {NoStop}%
\bibitem [{\citenamefont {Giovannetti}\ \emph {et~al.}(2008)\citenamefont
  {Giovannetti}, \citenamefont {Lloyd},\ and\ \citenamefont {Maccone}}]{QRAM}%
  \BibitemOpen
  \bibfield  {author} {\bibinfo {author} {\bibfnamefont {V.}~\bibnamefont
  {Giovannetti}}, \bibinfo {author} {\bibfnamefont {S.}~\bibnamefont {Lloyd}},
  \ and\ \bibinfo {author} {\bibfnamefont {L.}~\bibnamefont {Maccone}},\ }\href
  {\doibase 10.1103/PhysRevLett.100.160501} {\bibfield  {journal} {\bibinfo
  {journal} {Phys. Rev. Lett.}\ }\textbf {\bibinfo {volume} {100}},\ \bibinfo
  {pages} {160501} (\bibinfo {year} {2008})}\BibitemShut {NoStop}%
\bibitem [{\citenamefont {Mitarai}\ \emph {et~al.}(2019)\citenamefont
  {Mitarai}, \citenamefont {Kitagawa},\ and\ \citenamefont {Fujii}}]{CR}%
  \BibitemOpen
  \bibfield  {author} {\bibinfo {author} {\bibfnamefont {K.}~\bibnamefont
  {Mitarai}}, \bibinfo {author} {\bibfnamefont {M.}~\bibnamefont {Kitagawa}}, \
  and\ \bibinfo {author} {\bibfnamefont {K.}~\bibnamefont {Fujii}},\ }\href
  {\doibase 10.1103/PhysRevA.99.012301} {\bibfield  {journal} {\bibinfo
  {journal} {Phys. Rev. A}\ }\textbf {\bibinfo {volume} {99}},\ \bibinfo
  {pages} {012301} (\bibinfo {year} {2019})}\BibitemShut {NoStop}%
\bibitem [{\citenamefont {Grover}(2005)}]{fix2005}%
  \BibitemOpen
  \bibfield  {author} {\bibinfo {author} {\bibfnamefont {L.~K.}\ \bibnamefont
  {Grover}},\ }\href {\doibase 10.1103/PhysRevLett.95.150501} {\bibfield
  {journal} {\bibinfo  {journal} {Phys. Rev. Lett.}\ }\textbf {\bibinfo
  {volume} {95}},\ \bibinfo {pages} {150501} (\bibinfo {year}
  {2005})}\BibitemShut {NoStop}%
\bibitem [{\citenamefont {Yoder}\ \emph {et~al.}(2014)\citenamefont {Yoder},
  \citenamefont {Low},\ and\ \citenamefont {Chuang}}]{fix2014}%
  \BibitemOpen
  \bibfield  {author} {\bibinfo {author} {\bibfnamefont {T.~J.}\ \bibnamefont
  {Yoder}}, \bibinfo {author} {\bibfnamefont {G.~H.}\ \bibnamefont {Low}}, \
  and\ \bibinfo {author} {\bibfnamefont {I.~L.}\ \bibnamefont {Chuang}},\
  }\href {\doibase 10.1103/PhysRevLett.113.210501} {\bibfield  {journal}
  {\bibinfo  {journal} {Phys. Rev. Lett.}\ }\textbf {\bibinfo {volume} {113}},\
  \bibinfo {pages} {210501} (\bibinfo {year} {2014})}\BibitemShut {NoStop}%
\bibitem [{\citenamefont {Chakraborty}\ \emph
  {et~al.}(2018{\natexlab{b}})\citenamefont {Chakraborty}, \citenamefont
  {Gily{\'e}n},\ and\ \citenamefont {Jeffery}}]{CGJ1}%
  \BibitemOpen
  \bibfield  {author} {\bibinfo {author} {\bibfnamefont {S.}~\bibnamefont
  {Chakraborty}}, \bibinfo {author} {\bibfnamefont {A.}~\bibnamefont
  {Gily{\'e}n}}, \ and\ \bibinfo {author} {\bibfnamefont {S.}~\bibnamefont
  {Jeffery}},\ }\href@noop {} {\bibfield  {journal} {\bibinfo  {journal}
  {arXiv: quant-ph/1806.01838v1}\ } (\bibinfo {year}
  {2018}{\natexlab{b}})}\BibitemShut {NoStop}%
\bibitem [{\citenamefont {Ahuja}\ and\ \citenamefont
  {Kapoor}(1999)}]{ahuja1999}%
  \BibitemOpen
  \bibfield  {author} {\bibinfo {author} {\bibfnamefont {A.}~\bibnamefont
  {Ahuja}}\ and\ \bibinfo {author} {\bibfnamefont {S.}~\bibnamefont {Kapoor}},\
  }\href@noop {} {\bibfield  {journal} {\bibinfo  {journal} {arXiv:
  quant-ph/9911082}\ } (\bibinfo {year} {1999})}\BibitemShut {NoStop}%
\bibitem [{\citenamefont {Han}\ \emph {et~al.}(2012)\citenamefont {Han},
  \citenamefont {Kamber},\ and\ \citenamefont {Pei}}]{HAN201283}%
  \BibitemOpen
  \bibfield  {author} {\bibinfo {author} {\bibfnamefont {J.}~\bibnamefont
  {Han}}, \bibinfo {author} {\bibfnamefont {M.}~\bibnamefont {Kamber}}, \ and\
  \bibinfo {author} {\bibfnamefont {J.}~\bibnamefont {Pei}},\ }in\ \href
  {\doibase https://doi.org/10.1016/B978-0-12-381479-1.00003-4} {\emph
  {\bibinfo {booktitle} {Data Mining (Third Edition)}}},\ \bibinfo {series and
  number} {The Morgan Kaufmann Series in Data Management Systems},\ \bibinfo
  {editor} {edited by\ \bibinfo {editor} {\bibfnamefont {J.}~\bibnamefont
  {Han}}, \bibinfo {editor} {\bibfnamefont {M.}~\bibnamefont {Kamber}}, \ and\
  \bibinfo {editor} {\bibfnamefont {J.}~\bibnamefont {Pei}}}\ (\bibinfo
  {publisher} {Morgan Kaufmann},\ \bibinfo {address} {Boston},\ \bibinfo {year}
  {2012})\ \bibinfo {edition} {third edition}\ ed.,\ pp.\ \bibinfo {pages}
  {83--124}\BibitemShut {NoStop}%
\end{thebibliography}%

\end{document}